\documentclass[10pt,journal,draftclsnofoot,onecolumn]{IEEEtran}

\usepackage{times}
\usepackage{algorithmic}
\usepackage{amsmath}
\usepackage{amssymb}
\usepackage{amsthm}
\usepackage{color}
\usepackage{colortbl}
\usepackage{float}
\usepackage{cite}
\usepackage{subfig}
\usepackage{verbatim}
\usepackage[pdftex]{graphicx}

\usepackage[top=.9in, bottom=.9in, left=.80in, right=.80in]{geometry}

\usepackage{slashbox}

\newtheorem{theorem}{Theorem}
\newtheorem{lemma}[theorem]{Lemma}

\newtheorem{corollary}[theorem]{Corollary}

\newtheorem{example}{Example}
\newtheorem{remark}{Remark}

\newcommand{\deff}{\mbox{$\stackrel{\rm def}{=}$}}
\newcommand{\field}[1]{\mathbb{#1}}

\newcommand{\F}{\field{F}}

\newcommand{\cB}{{\mathcal B}}
\newcommand{\cC}{{\mathcal C}}
\newcommand{\cD}{{\mathcal D}}
\newcommand{\cG}{{\mathcal G}}

\newcommand{\cP}{{\mathcal P}}

\newcommand{\sP}{\cP}
\newcommand{\sG}{\cG}

\newcommand{\Gr}{\smash{{\sG\kern-1.5pt}_q\kern-0.5pt(n,k)}}
\newcommand{\Gk}{\smash{{\sG\kern-1.5pt}_q\kern-0.5pt(n,k_1)}}
\newcommand{\Gkk}{\smash{{\sG\kern-1.5pt}_q\kern-0.5pt(n,k_2)}}
\newcommand{\Grtwo}{\smash{{\sG\kern-1.5pt}_2\kern-0.5pt(n,k)}}
\newcommand{\Gkone}{\smash{{\sG\kern-1.5pt}_q\kern-0.5pt(n,k_1)}}
\newcommand{\Gktwo}{\smash{{\sG\kern-1.5pt}_q\kern-0.5pt(n,k_2)}}
\newcommand{\Ps}{\smash{{\sP\kern-2.0pt}_q\kern-0.5pt(n)}}

\begin{document}
\title{Optimal Fractional Repetition Codes\\ based on Graphs and Designs}

\author{Natalia Silberstein and Tuvi Etzion
\thanks{The authors are with the Department of Computer Science,
  Technion -- Israel Institute of Technology, Haifa 32000, Israel (email: natalys@cs.technion.ac.il, etzion@cs.technion.ac.il).
   }
    \thanks{This research was supported in part  by the Israeli Science Foundation (ISF), Jerusalem, Israel, under Grant 10/12.}
   \thanks{The first author was supported in part at the Technion by a Fine Fellowship.}
   \thanks{This work was presented in part at the Information Theory and Applications Workshop (ITA 2014), San-Diego, USA, February 2014.
   }
}

\maketitle

\begin{abstract}
Fractional repetition (FR) codes is a family of codes for distributed storage systems
that allow for uncoded exact repairs having the minimum repair bandwidth. However, in contrast
to minimum bandwidth regenerating (MBR) codes, where a random set of a certain size of
available nodes is used for a node repair, the repairs with FR codes are table based. This usually allows to store more data compared to MBR codes.
In this work, we consider bounds on
the fractional repetition capacity, which is the maximum amount of data that can be stored using an FR code.
Optimal
FR codes which attain these bounds are
presented. The constructions of these
FR codes are based on combinatorial designs and
on families of regular and biregular graphs.
These constructions of FR codes for given parameters raise some interesting questions
in graph theory. These questions and some of their solutions are discussed in this paper.
In addition, based on a connection between FR codes and batch codes, we propose a new family of codes for DSS, namely fractional repetition batch codes, which have the properties of batch codes and FR codes simultaneously.  These are the first codes for DSS which allow for uncoded efficient exact repairs and load balancing which can be performed by several users in parallel.
Other concepts related to FR codes are also discussed.
\end{abstract}




\section{Introduction}
\label{sec:introduction}
In a distributed storage system (DSS), data is stored across a network of nodes, which can unexpectedly fail. To provide reliability, data redundancy based on coding techniques is introduced in such systems. Moreover, existing erasure codes allow to minimize  the storage overhead~\cite{WeKu02}.
Dimakis et al.~\cite{dimakis} introduced a new family of erasure codes, called \emph{regenerating codes}, which allow for efficient single node repairs by minimizing  repair bandwidth. In particular, they presented two families of regenerating codes, called \emph{minimum storage regenerating} (MSR) and \emph{minimum bandwidth regenerating} (MBR) codes, which correspond to the two extreme points on the storage-bandwidth trade-off~\cite{dimakis}. Constructions for these
two families of codes can be found in~\cite{bruck11,dimakis,DRWS11,Rashmi09,rashmi,shah,suh1} and references therein.

An $(n,k,d,M,\alpha,\beta)_q$ regenerating code $C$, for $k\leq d\leq n-1$, $\beta\leq \alpha$, is used to store a file of size $M$ across a network of $n$ nodes, where each node stores $\alpha$ symbols from $\F_q$, a finite field with $q$ elements, such that the stored file can be recovered by downloading  the data from any set of $k$ nodes,  where $k$ is called the \emph{reconstruction degree}. Note, that this means that any $n-k$ node failures (i.e., erasures) can be corrected by this code. When a single node fails, a newcomer node which substitutes the failed node contacts  any set of $d$ nodes and downloads $\beta$ symbols of each node in this set to reconstruct the failed data.   This process is called a \emph{node repair process}, and the amount of data downloaded to repair a failed node, $\beta d$, is called the \emph{repair bandwidth}.

The family of MBR codes has the minimum possible repair bandwidth, namely $\beta d=\alpha$.
In~\cite{Rashmi09,RSKR12} Rashmi et al. presented a construction for MBR codes which have the additional property of exact \emph{repair by transfer}, or exact \emph{uncoded repair}. In other words, the code proposed in~\cite{Rashmi09,RSKR12} allows for efficient  node repairs where no decoding is needed. Every node participating in a node repair process just passes one  symbol
($\beta=1$) which will be  directly stored in the newcomer node. This construction is based on a concatenation of an outer MDS code with an inner repetition code based on a complete graph as follows.
Let  $M=k\alpha -\binom{k}{2}$  be the size of a file, which corresponds to MBR capacity with $\beta=1$~\cite{dimakis}. This file is first encoded by using an $\left(\binom{n}{2},M\right)$ MDS code $\cC$. The $\binom{n}{2}$ symbols of the corresponding codeword of $\cC$ are placed on the $n$ different nodes, where each node stores $\alpha=n-1$ symbols, as follows. Each node is associated with a vertex in  $K_n$, the complete graph with $n$ vertices. Every symbol of the codeword from $\cC$ is associated with an edge of $K_n$. Each node $i$ of the DSS stores the symbols of the codeword of $\cC$ which are associated with the edges incident to vertex $i$ of $K_n$.
The uniqueness of this construction for the given parameters $\alpha=d=n-1$ was proved in~\cite{Rashmi09}.

El Rouayheb and Ramchandran~\cite{RoRa10} generalized the construction of~\cite{Rashmi09} and defined a new family of codes for DSSs which also allow exact repairs by transfer for a wide range of parameters. These codes, called \emph{DRESS} (Distributed Replication based Exact Simple Storage) codes~\cite{PNRR11}, consist of the concatenation of an outer MDS code and an inner repetition code called \emph{fractional repetition} (FR) code.
These codes for DSSs relax the requirement of a random $d$-set for a repair of a failed node and instead the repair becomes table based (or by using an appropriate function). This modified model requires new bounds on the maximum amount of data that can be stored on a DSS based on an FR code.

An $(n,\alpha, \rho)$ FR code $C^{\textmd{FR}}$ 
is a collection of $n$ subsets $N_1,\ldots, N_n$ of $[\theta]\deff\{1,2,\ldots,\theta\}$, $n\alpha=\rho \theta$,  such that
\begin{itemize}
  \item $|N_i|=\alpha$ for  each $i$, $1\leq i\leq n$;
  \item each symbol of $[\theta]$ belongs to exactly $\rho$ subsets in $C^{\textmd{FR}}$, where $\rho$ is called the \emph{repetition degree} of $C^{\textmd{FR}}$.
\end{itemize}

A $\left[(\theta, M), k, (n,\alpha, \rho)\right]$ DRESS code is a code obtained by the concatenation of an outer $(\theta, M)$ MDS code $\cC$
and an inner $(n,\alpha, \rho)$ FR code $C^{\textmd{FR}}$. To store a file $\textbf{f}\in \F_q^M$  on a DSS, $\textbf{f}$  is first encoded by using
$\cC$; next, the $\theta$ symbols of the codeword $\textbf{c}_{\textbf{f}}\in \cC$, which encodes the file $\textbf{f}$, are placed on the $n$ nodes defined by $C^{\textmd{FR}}$ as follows: node $i\in [n]$ of the DSS stores $\alpha$ symbols of $\textbf{c}_\textbf{f}$, indexed by the elements of the subset $N_i$.
Each symbol of $\textbf{c}_\textbf{f}$ is stored in exactly $\rho$ nodes  and it is possible to reconstruct the stored file $\textbf{f}$ from any set of $k$ nodes.
When some node $j$ fails, it can be repaired by using  a set of $d=\alpha$ other nodes  $i_1,i_2\ldots,i_{\alpha}$, such that
$N_j\cap N_{i_s}\neq \varnothing$, $s\in [\alpha]$, and $\cup_{s\in [\alpha]}(N_j\cap N_{i_s})=N_j$. Each such node passes exactly one symbol ($\beta=1$) to repair node $j$. Note that the repair bandwidth of a DRESS code is the same as the repair bandwidth of an MBR code.
The encoding scheme based on an FR code is shown in Fig.~\ref{fig:FRscheme}.

\begin{figure*}[t]
\centering
\includegraphics[trim=0 0 0 70,clip=true,width=0.6\textwidth]{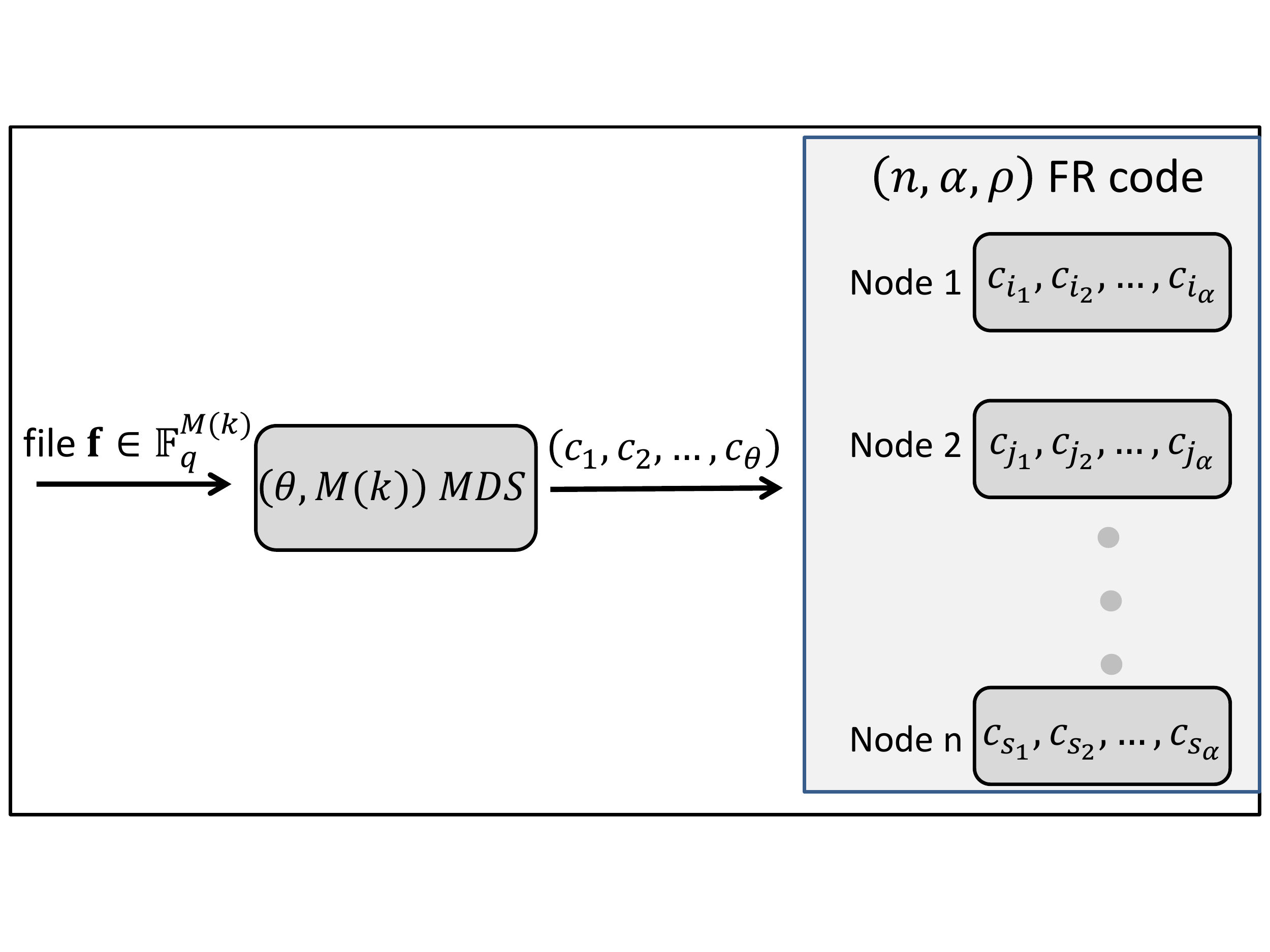}
\caption{The encoding scheme for a DRESS code}\label{fig:FRscheme}
\end{figure*}

Note that the stored file should be reconstructed
from any set of $k$ nodes, and since the outer code is an MDS code of dimension $M$,
it follows that
\begin{equation}
\label{eq:rate}
M\leq \min_{|I|=k}|\cup_{i\in I}N_i|,
\end{equation}
where at least $\min_{|I|=k}|\cup_{i\in I}N_i|$ of distinct symbols of the MDS codeword are contained in any set of $k$ nodes.
Since we want to maximize the size of a file that can be stored by using  a DRESS code, in the sequel we will always assume that $M=\min_{|I|=k}|\cup_{i\in I}N_i|$\footnote{In some works this value is called the \emph{rate }of a code. To avoid a confusion with a rate of a classical code we refer to this value as to the \emph{maximum file size}.}.
Note, that the same FR code can be used in different DRESS codes, with different $k$'s as reconstruction degrees, and different MDS codes. The file size $M$, which is the dimension of the chosen MDS code,  depends on the value of chosen $k$ and hence in the sequel we will use $M(k)$ to denote the size of the file.

An $(n,\alpha, \rho)$  FR
code is called \emph{universally good}~\cite{RoRa10} if for any $k\leq \alpha$ the $\left[(\theta, M(k)), k, (n,\alpha, \rho)\right]$ DRESS
code satisfies
\begin{equation}
\label{eq:univGood}
M(k)\geq k\alpha-\binom{k}{2},
\end{equation}
where the  righthand side of equation~(\ref{eq:univGood}) is the maximum file size that can be stored using an MBR code, i.e., the MBR capacity~\cite{dimakis}.
In particular, it is of interest to consider codes which allow to store larger files  when compared to MBR codes.
 Note that to satisfy~(\ref{eq:univGood}) the inner FR code of a DRESS code should satisfy that
\begin{equation}
\label{eq:FRunGood}
\min_{|I|=k}|\cup_{i\in I}N_i|\geq k\alpha-\binom{k}{2}.
\end{equation}
  Note also that if an FR code $C^{\textmd{FR}}$ satisfies~(\ref{eq:FRunGood}) then $|N_i\cap N_j|\leq 1$, for $N_i,N_j\in C^{\textmd{FR}}$, $i\neq j\in[n]$~\cite{Rashmi09}.

Two upper bounds on the maximum file size $M(k)$ of a $\left[(\theta, M(k)), k, (n,\alpha, \rho)\right]$ DRESS code ($n\alpha=\rho\theta$), called the \emph{FR capacity} and denoted in the sequel by $A(n,k,\alpha, \rho)$, were presented in~\cite{RoRa10}:
  \begin{equation}\label{eq:bound1}
  A(n,k,\alpha, \rho)\leq \left\lfloor\frac{n\alpha}{\rho}\left(1-\frac{\binom{n-\rho}{k}}{\binom{n}{k}}\right)\right\rfloor;
  \end{equation}
  \begin{equation}\label{eq:bound2}
   A(n,k,\alpha, \rho)\leq  \varphi(k), \textmd{ where } \varphi(1)=\alpha,\textmd{  } \varphi(k+1)=\varphi(k)+\alpha-\left\lceil\frac{\rho \varphi(k)-k\alpha}{n-k}\right\rceil.
  \end{equation}

  Note, that the bound in~(\ref{eq:bound2}) is tighter than bound in~(\ref{eq:bound1}).
   Note also that for any given $k$, the function $A(n,k,\alpha, \rho)$ is determined by the parameters of the inner FR code.
   We call an FR code \emph{$k$-optimal} if it satisfies
 $$\min_{|I|=k}|\cup_{i\in I}N_i|=A(n,k,\alpha,\rho),$$
  in other words, the size of a file stored by using the FR code is the maximum possible for the given $k$.
   We call an FR code \emph{optimal} if  for any $k\leq \alpha$  it is $k$-optimal
 \footnote{Note that since $\alpha=d$ we always have $k\leq d=\alpha$ (very similarly to the parameters of regenerating codes).}.

Constructions for FR codes are considered in many papers starting from~\cite{RoRa10}, where FR codes are constructed from random regular graphs, Steiner systems, and their dual designs. Randomized FR codes based on the balls-and-bins model are presented in~\cite{PNRR11}. Constructions of FR codes with the fewest number of storage nodes given the other parameters, based on finite geometries and corresponding  bipartite cage graphs are considered in~\cite{KoGi11}. FR codes based on affine resolvable designs and mutually orthogonal Latin squares are presented in~\cite{OlRa12,OlRa14}.  Enumeration of FR codes up to a given number of system nodes is presented in~\cite{AGG13}. Algorithms for computing the reconstruction degree $k$ and the repair degree $d$ of FR codes are presented in~\cite{BGA13}. Construction of FR codes based on regular graphs with a given girth, in particular cages,  and analysis of their minimum distance is considered in~\cite{OlRa13}. Generalization of FR codes to \emph{weak} or \emph{general} FR codes, where each node stores a different amount of symbols, and constructions of codes based on graphs and group divisible designs are considered in~\cite{GAY13,ZSLH14}.

Note, that in all these papers the optimality of the constructed FR codes regarding the FR capacity, i.e. the maximality of the size of the stored file, was not considered.
In this paper,  we address the problem of constructing $k$-optimal FR codes  and optimal FR codes (and hence optimal DRESS codes). In addition, we consider various problems from graph theory raised from the problem of constructing FR codes and present FR codes with additional desired properties.
%
%

The rest of the paper is organized as follows. In Section~\ref{sec:preliminaries}
we provide the main definitions of the structures which will be used in our constructions.
In particular, in Subsection~\ref{subsec:graphs}  we provide definitions for some families of regular graphs and
graphs with a given girth. We present the Tur\'an's theorem and the Moore bound which are essential for the  results in this paper.
In Subsection~\ref{subsec:designs} we provide the definitions of
transversal designs, projective planes, generalized polygons, and their incidence matrices.
The definitions of FR codes based on graphs or on designs are given in Subsection~\ref{subsec:codes}.

In Section~\ref{sec:rho2}
we consider FR codes with $\rho=2$
and  propose  constructions for FR codes  which
attain the bound in~(\ref{eq:bound2}).
 Some of these codes are optimal and some  are $k$-optimal for specific values of $k$. Note, that the case $\rho=2$
corresponds to the case of the highest data/storage ratio,
since the repetition degree is the lowest one. All the
constructions in this section are based on different families of regular graphs.
First, we provide a useful lemma which shows the connection between the file size of a code and the structure of its underlying graph.
In Subsection~\ref{subsec:constructions_rho2}  optimal FR codes
based on  Tur\'an graphs are considered.
In Subsection~\ref{subsec:k-optimal} $k$-optimal FR codes based on different regular graphs with a given girth are presented.
In Subsection~\ref{subsec:file_size} FR codes with a given file size are considered. The constructions raise many interesting questions in graph theory which are discussed in this section.

 In Section~\ref{sec:rho>2} we consider FR codes with
$\rho>2$. In this case, a failed node can  be repaired from several sets of other nodes, in contrast to the case with $\rho=2$, in which a failed node can be repaired from a unique subset of $\alpha$ available nodes. One construction is based on a family of combinatorial designs, called transversal designs. This construction generalizes the construction based on Tur\'an graphs for $\rho=2$. Another construction is based on biregular bipartite graphs with a given girth. One important family of such graphs are the generalized polygons. We analyze the parameters of the constructed codes and find the conditions for which the bound in~(\ref{eq:bound2}) is attained.

In Section~\ref{sec:genHamming} we establish a connection between the file size hierarchy  and the generalized Hamming weight hierarchy.
In fact, the sizes of the
file related to the increasing values of $k$'s form an integer sequence of nondecreasing
values which can be viewed as a generalized definition for Hamming weights
for constant weight codes.
In Section~\ref{sec:bound_k} we provide a lower bound on the reconstruction degree and present some FR codes which attain this bound.

In Section~\ref{sec:batch} we analyze additional properties of  FR codes by establishing  a connection between FR codes and \emph{combinatorial batch codes}.
We propose a novel family of codes for DSS, called \emph{fractional repetition batch codes} (FRB), which enable exact uncoded repairs and load balancing that can be performed by several users in parallel. We present examples of constructions of FRB codes based on bipartite complete graphs, graphs with large girth, transversal designs and affine planes.
Conclusion are given in Section~\ref{sec:conclusion}.


\section{Preliminaries}
\label{sec:preliminaries}
In this section we provide the definitions of all the combinatorial objects used for the constructions of FR codes presented in this paper.

\subsection{Regular and Biregular Graphs}
\label{subsec:graphs}
A graph $G=(V,E)$  consists of a vertex set $V$ and an edge set $E$, where an edge is an unordered pair of vertices of~$V$. For an edge $e=\{x,y\}\in E$ we say that $x$ and $y$ are \emph{adjacent} and that $x$ and $e$ are \emph{incident}. The \emph{degree} of a vertex $x$ is the number of edges incident with it. We say that a graph $G$ is \emph{\emph{regular}} if all its vertices have the same degree and $G$ is $d$-\emph{regular} if each vertex has degree $d$. A graph is called \emph{connected} if there is path between any pair of vertices.
A graph is called \emph{complete} if every pair of vertices are adjacent. A complete graph on $n$ vertices is denoted by $K_n$. A \emph{subgraph} $G_2=(V_2,E_2)$ of a graph $G_1=(V_1,E_1)$ is a graph  such that $V_2\subseteq V_1$ and $E_2\subseteq E_1$. A subgraph $G_2$ of a graph $G_1$ is called \emph{induced}  if $E_2=\{\{x,y\}:x,y\in V_2, \{x,y\}\in E_1\}$.
A \emph{k-clique} in a graph $G$ is a complete subgraph of $G$ with $k$ vertices.
The \emph{complement} of a graph $G=(V,E)$, denoted by $\overline{G}= (V, \overline{E})$, is a graph with the same vertex set $V$ but whose edge set $\overline{E}$ consists of all the edges not contained in $G$, i.e., $\overline{E}=\{\{x,y\}:x,y \in V, x\neq y, \{x,y\}\notin E\}$.

The \emph{incidence matrix }$\textbf{I}(G)$ of a graph $G=(V,E)$ is a binary $|V|\times |E|$ matrix with rows and columns indexed by the vertices and edges of $G$, respectively, such that $(\textbf{I}(G))_{i,j}=1$ if and only if vertex $i$ and edge $j$ are incident.

A graph $G$ is called \emph{bipartite} ($r$-\emph{partite}, respectively) if its vertex set can be partitioned into two ($r$, respectively) parts such that every two adjacent vertices belong to two different parts.
 A bipartite graph is denoted by $G=(L\cup R, E)$, where $L$ is the left part and $R$ is the right part of $G$.
 A bipartite graph $G$ is called \emph{biregular} if the degree of the vertices in one part is $d_1$ and the degree of the vertices in the other part is $d_2$.
 An $r$-partite graph is called \emph{complete} if every two vertices from two different parts are connected by an edge.
The complete bipartite graph with left part of size $n$ and right part of size $m$ is denoted by $K_{n,m}$.
Note that in $K_{n,m}$ the degree of a vertex in the left part is $m$ and the degree of a vertex in the right part is $n$.

The following theorem, known as \emph{Tur\'an's theorem},
provides a necessary condition that a graph
does not contain a clique of a given size~\cite[p. 58]{Jukna}.
\begin{theorem}
\label{thm:turan}
If a graph $G = (V,E)$ on $n$ vertices has no $(r+1)$-clique, $r \geq 2$,
then
\begin{equation}
\label{turanTheorem}
 |E| \leq(1-\frac{1}{r})\frac{n^2}{2}.
\end{equation}
\end{theorem}

\begin{corollary}
\label{cor:Turan}
If $G$ is an $\alpha$-regular graph which does not contain an $(r+1)$-clique then
$
n\geq \frac{r}{r-1}\alpha.
$
\end{corollary}

We consider a family of regular graphs, called \emph{Tur\'an graphs},  which attain the bound of Corollary~\ref{cor:Turan}, in other words,  have the smallest number of vertices.
Let $r,n$ be two integers such that $r$ divides $n$. An $(n,r)$-\emph{Tur\'an graph} is defined as a regular complete $r$-partite graph, i.e., a graph formed by partitioning a set of $n$ vertices into $r$ parts of size $\frac{n}{r}$ and connecting each two vertices of different parts by an edge.  Clearly, an $(n,r)$-Tur\'an graph does not contain a clique of size $r+1$  and it is an $(r-1)\frac{n}{r}$-regular graph.

We now turn to another family of graphs, called \emph{cages}.
A \emph{cycle} in a graph $G$ is a connected subgraph of $G$ in which each vertex has degree two. The \emph{girth} of a graph is the length of its shortest cycle.
A $(d,g)$-\emph{cage} is a $d$-regular graph with girth $g$ and minimum number of vertices. For example, a $(d,4)$-cage is a complete bipartite graph $K_{d,d}$. Constructions for cages are known for $g\leq 12$~\cite{ExJa08}. Let $\widehat{n}_0(d,g)$ be the minimum number of vertices in a $(d,g)$-cage. A lower bound on the number of vertices in a $(d,g)$-cage is given in the  following theorem, known as \emph{Moore bound}~\cite[p. 180]{Biggs}.

\begin{theorem}
\label{MooreBound}
The number of vertices in a $(d,g)$-cage is at least
\begin{equation}
\label{eq:cage}
   n_0(d,g)=\left\{\begin{array}{c c}
         1+d\sum_{i=0}^{\frac{g-3}{2}}(d-1)^i &\;\textmd{ if }g \textmd{ is odd} \\
            2\sum_{i=0}^{\frac{g-2}{2}}(d-1)^i&\;\textmd{ if }g \textmd{ is even}
         \end{array}\right..
       \end{equation}
\end{theorem}

Similar result to the Moore bound for biregular bipartite graphs can be found in~\cite{ABV09,FLSUW95}.


\subsection{Combinatorial Designs}
\label{subsec:designs}

A \emph{set system} is a pair $(\cP,\cB)$, where $\cP=\{p_i\}$ is a finite nonempty set of \emph{points} and  $\cB=\{B_i\}$ is a finite
nonempty set of subsets of $\cP$ called \emph{blocks}. A \emph{design} $\cD$ is set system with a constant number of points per block and no
repeated blocks.  A design $\cD$ can be described by an \emph{incidence matrix} $\textbf{I}(\cD)$, which is a binary $|\cP|\times|\cB|$ matrix, with rows indexed by the points, columns indexed by the blocks, where
\begin{equation*}
\left.\begin{array}{c}
        (\textbf{I}(\cD))_{i,j}=\left\{\begin{array}{cc}
         1 &\; \textmd{ if }p_i\in B_j\\
            0&\; \textmd{ if }p_i\notin B_j
         \end{array}\right.
         \end{array}\right..
\end{equation*}

The \emph{incidence} graph $G_I(\cD)=(V,E)$ of $\cD$  is  the bipartite graph with the vertex set $V=\cP\cup\cB$, where  $\{p,B\}\in E$
if and only if $p\in B$, for $p\in \cP$, $B\in\cB$.

A \emph{transversal design} of group size $h$ and block size $\ell$,  denoted by $\text{TD}(\ell, h)$
is a triple $(\cP,\mathcal{G},\mathcal{B})$, where

\begin{enumerate}
\item $\cP$ is a set of $\ell h$ \emph{points};

\item $\mathcal{G}$ is a partition of $\cP$ into $\ell$ sets
(\emph{groups}), each one of size $h$;

\item $\mathcal{B}$ is a collection of $\ell$-subsets of $\cP$
(\emph{blocks});

\item each block meets each group in exactly one point;

\item any pair of points from different groups is contained in exactly one block.
\end{enumerate}


 The properties of a transversal design $\text{TD}(\ell,h)$ which will be useful for our constructions  are summarized in the following lemma~\cite{Anderson}.

\begin{lemma}
\label{lm:TDparameters}
Let $(\cP,\cG,\cB)$ be a transversal design $\text{TD}(\ell,h)$. Then
\begin{itemize}
  \item The number of points is given by $|\cP|=\ell h$;
  \item The number of groups is given by $|\cG|=\ell$;
  \item The number of blocks is given by $|\cB|=h^2$;
  \item The number of blocks that contain a given point is equal to $h$.
  \item The girth of the incidence graph of a transversal design is equal to $6$.
\end{itemize}
\end{lemma}


A $\text{TD}(\ell,h)$ is called \emph{resolvable} if the set $\mathcal{B}$
can be partitioned into subsets $\mathcal{B}_1,...,\mathcal{B}_h$, each one contains $h$ blocks,
such that each element of $\cP$ is contained in exactly one block of
each $\mathcal{B}_i$, i.e., the blocks of $\cB_i$ partition the set $\cP$.
Resolvable transversal design $\text{TD}(\ell,q)$ is known to exist for any $\ell \leq q$ and prime power $q$~\cite{Anderson}.


\begin{remark}A $\textmd{TD}(2,h)$, for any integer $h\geq 2$ is equivalent to the complete bipartite graph $K_{h,h}$.
\end{remark}

Next, we consider  two families of designs whose incidence graphs  attain the Moore bound~(\ref{eq:cage}).

A \emph{projective plane} of order $n$ denoted by $\textmd{PG}(2,n)$, is a design $(\cP,\cB)$, such that $|\cP|=|\cB|=n^2+n+1$, each block of $\cB$ is of size $n+1$, and any two points are contained in exactly one block. Note that any two blocks in $\cB$ have exactly one common point. It is well known (see~\cite{Godsil}) that the incidence graph of a projective plane has girth $6$.

A \emph{generalized quadrangle} of order $(s,t)$, denoted by $\textmd{GQ}(s,t)$ is a design $(\cP,\cB)$, where
\begin{itemize}
  \item Each point $p\in \cP$ is incident with $t+1$ blocks, and each block $B\in\cB$ is incident with $s+1$ points.
  \item Any two blocks have at most one common point.
  \item For any  pair $(p,B)\in \cP\times\cB$, such that $p\notin B$, there is exactly one block $B'$
 incident with $p$, such that $|B'\cap B|=1$.
\end{itemize}
In a generalized quadrangle $\textmd{GQ}(s,t)$, the number of points $|\cP|=(s+1)(st+1)$, the number of blocks $|\cB|=(t+1)(st+1)$ and the girth of the incidence graph is $8$~\cite{Godsil}.

We note that transversal designs, projective planes, and generalized quadrangles belong to a class of  designs called \emph{partial geometries}.
In addition, projective planes and generalized quadrangles are examples of designs called \emph{generalized polygons} (or $n$-\emph{gons}). Their incidence graphs have girth $2n$ and they attain the Moore bound. Such structures are known to exist only for $n\in\{3,4,6,8\}$~\cite{Godsil}.

\subsection{FR Codes based on Graphs and Designs }
\label{subsec:codes}

Let $C$ be an $(n,\alpha,\rho)$ FR code. $C$ can be described by an \emph{incidence matrix} $\textbf{I}(C)$, which is an $n\times \theta$ binary matrix, $\theta=\frac{n\alpha}{\rho}$, with rows indexed by the nodes of the code and columns indexed by the symbols of the corresponding MDS codeword, such that $(\textbf{I}(C))_{i,j}=1$ if and only if node $i$ contains symbol $j$.

Let $G$ be an $\alpha$-regular graph with $n$ vertices. We say that an $(n,\alpha,\rho=2)$ FR code $C$ is based on $G$ if $\textbf{I}(C)=\textbf{I}(G)$. Such a code will be denoted by $C_G$.
It can be readily verified that any $(n,\alpha,2)$ FR code can be represented by an $\alpha$-regular graph with $n$ vertices.

Let $\cD=(\cP,\cB)$ be a design with $|\cP|=n$ points such that each block $B\in\cB$ contains $\rho$ points and each point $p\in \cP$ is contained in $\alpha$ blocks. We say that an $(n,\alpha,\rho)$ FR code $C$ is based on $\cD$ if $\textbf{I}(C)=\textbf{I}(\cD)$. Such a code will be denoted by $C_{\cD}$.


\section{Fractional Repetition Codes with Repetition Degree 2}
\label{sec:rho2}

In this section we present constructions of optimal and $k$-optimal FR codes with repetition degree  $\rho=2$.  These constructions are based on different types of regular graphs and are given in Subsections~\ref{subsec:constructions_rho2} and ~\ref{subsec:k-optimal}. In Subsection~\ref{subsec:file_size}  the properties of these graphs are investigated in order to present FR codes which allow to store a file of any given size. To avoid triviality we assume throughout the section that $\alpha>2$.

First, we present the following useful lemma which shows a connection between the problem of finding the file size of an FR code based on a graph and the edge isoperimetric problem on graphs~\cite{Bez99}.

\begin{lemma}
\label{lm:isoperimetric}
Let $G=(V,E)$ be an $\alpha$-regular graph and let $C_G$ be the FR code based on $G$. We denote by $G_k$ the family of induced subgraphs of $G$ with $k$ vertices, i.e.,
$$G_k=\{G'=(V',E'):|V'|=k, G'\textmd{ is an induced subgraph of }G\}.$$
Then the file size $M(k)$ of $C_G$ is given by
$$M(k)=k\alpha -\max_{G'\in G_k}|E'|.
$$
\end{lemma}

\begin{proof}
For each induced subgraph $G'=(V',E')\in G_k$ we define  $E'_{\textmd{cut}}$ to be the set of all the edges of $E$ in the cut between $V'$ and $V\setminus V'$, i.e.,
$$E'_{\textmd{cut}}=\{\{v,u\}\in E: v\in V', u\in V\setminus V'\}.
$$
Clearly,  $k\alpha=2|E'|+|E'_{\textmd{cut}}|$ for every $G'\in G_k$.
Note that $M(k)=\min_{G'\in G_k}\{|E'|+|E'_{\textmd{cut}}|\}$ and hence
$$M(k)=\min_{G'\in G_k}\{|E'|+\alpha k-2|E'|\}=\alpha k-\max_{G'\in G_k}\{|E'|\}.
$$
\end{proof}

\vspace{-.2cm}
\subsection{Optimal FR Codes Based on Tur\'an Graphs }
\label{subsec:constructions_rho2}
We begin our discussion with
the following lemma which follows directly from Lemma~\ref{lm:isoperimetric}.

\begin{lemma}
\label{lm:clique} Let $G$ be an $\alpha$-regular graph with $n$ vertices, and let $M(k)$ be the file size of the corresponding FR code~$C_{G}$.
The graph $G$ contains a $k$-clique if and only if $M(k)=k\alpha-\binom{k}{2}$.
\end{lemma}

\begin{corollary}
 \label{cor:clique}
 The file size $M(k)$ of an FR code $C_G$, where $G$  is a graph which does not contain a $k$-clique, is strictly larger than the MBR capacity.
\end{corollary}

One of the main advantages of an FR code is that its file size usually exceeds the MBR capacity. Hence, as a consequence of Corollary~\ref{cor:clique}, we consider different families of regular graphs which do not contain a $k$-clique for a given $k$.
Therefore, since Tur\'an graphs have the minimum number of vertices among the graphs which
do not contain a  clique of a given size (see Corollary~\ref{cor:Turan}),  we consider FR codes based on Tur\'an graphs.
The following theorem shows that FR codes obtained from Tur\'an graphs attain the upper bound in~(\ref{eq:bound2}) for all $k\leq \alpha$ and hence they are optimal FR codes.

\begin{theorem}
\label{trm:Turan}
Let $T=(V,E)$ be an $(n,r)$-Tur\'an graph, $\alpha=(r-1)\frac{n}{r}$, and let $k$ be an integer such that $1\leq k\leq \alpha$. If $k=br+t$ for nonnegative integers $b,t$ such that $t\leq r-1$ then the $(n,\alpha,2)$ FR code $C_{T}$ based on $T$ has file size
\begin{equation}
\label{eq:turanRate1}
M(k)=k\alpha-\binom{k}{2}+r\binom{b}{2}+bt,
\end{equation}
 which attains the upper bound in~(\ref{eq:bound2}).
\end{theorem}

\begin{proof}
By Lemma~\ref{lm:isoperimetric}, the value of $M(k)$ is determined by
 maximum cardinality of the edge set in an induced subgraph $T'=(V',E')$ of $T$, where $|V'|=k$.
One can verify that
since $T$ is a complete regular $r$-partite graph, it follows that the induced subgraph $T'$ with  $E'$ of the  maximum cardinality
is a complete  $r$-partite graph with exactly $t$ parts of size $b+1$ and $r-t$ parts of size $b$.  Hence, the number of edges in $E'$ is given by
\begin{equation*}
\binom{t}{2}(b+1)^2+\binom{r-t}{2}b^2+t(r-t)(b+1)b.
\end{equation*}
Thus, by Lemma~\ref{lm:isoperimetric},
\begin{equation}
\label{eq:turanRate}
M(k)=k\alpha -\left[\binom{t}{2}(b+1)^2+\binom{r-t}{2}b^2+t(r-t)(b+1)b\right].
\end{equation}
It can  be easily  verified that~(\ref{eq:turanRate1}) equals to~(\ref{eq:turanRate}). In addition, one can verify (by induction) that for the parameters of the constructed code $C_T$ the bound in~(\ref{eq:bound2}) equals to~(\ref{eq:turanRate1}).
\end{proof}

\begin{remark}
Note, that for any  $k> r$, the file size of the code $C_{T}$  is strictly larger than the MBR capacity, i.e.,
$$M(k) > k\alpha-\binom{k}{2}.
$$
\end{remark}

In the following theorem we provide an alternative, simpler representation of a file size for the FR code based on a Tur\'an graph. Obviously, this expression for the file size is equivalent to the expression in~(\ref{eq:turanRate1}). The proof of this theorem is also simpler than the one of Theorem~\ref{trm:Turan}. However, the proof of Theorem~\ref{trm:Turan} could be used for the proof of Theorem~\ref{thm:batch_TD} in Section ~\ref{sec:rho>2} and hence it was given for completeness.

\begin{theorem}
\label{thm:alternativeTuran}
Let $T=(V,E)$ be an $(n,r)$-Tur\'an graph, $r<n$, $\alpha=(r-1)\frac{n}{r}$, and let $k$ be an integer such that $1\leq k\leq \alpha$. Then the $(n,\alpha,2)$ FR code $C_{T}$ based on $T$ has the file size given by
\[M(k)=k\alpha-\left\lfloor\frac{r-1}{r}\cdot\frac{k^2}{2}\right\rfloor.
\]
\end{theorem}

\begin{proof}
We consider an induced  subgraph $T'=(V',E')$ of the Tur\'an graph $T$ with $|V'|=k$ vertices which has the maximum number of edges. Since $T'$ is a subgraph of $T$, in particular it does not contain $K_{r+1}$. Then, by Tur\'an's theorem (see  Theorem~\ref{thm:turan}), $|E'|\leq \frac{r-1}{r}\frac{k^2}{2}$. Hence by Lemma~\ref{lm:isoperimetric}, assuming that $M(k)$ is an integer, we have
\[M(k)=k\alpha-\left\lfloor\frac{r-1}{r}\cdot\frac{k^2}{2}\right\rfloor.
\]
\end{proof}

\vspace{1,0cm}

The following result for FR codes based on complete bipartite graphs is a special case of Theorem~\ref{thm:alternativeTuran} with $r=2$.

\begin{corollary}
\label{thm:TD(2,a)}
The maximum size $M(k)$ of a file that can be stored using the $(2\alpha,\alpha, 2)$ FR code $C_{K_{\alpha,\alpha}}$ based on a regular complete bipartite graph $K_{\alpha,\alpha}$, for $\alpha\geq 2$, is given by

\begin{equation}\label{eq:rho2 file size}
M(k)=\left\{\begin{array}{cc}
         k\alpha -\frac{k^2}{4}&\;\textmd{ if }k\textmd{ is even} \\
            k\alpha -\frac{k^2-1}{4}&\;\textmd{ if }k\textmd{ is odd}
         \end{array}\right.
\end{equation}
which attains the upper bound in~(\ref{eq:bound2}) for all $1\leq k\leq \alpha$.
\end{corollary}

\begin{example}
\label{ex:bipartite}
 The $(6,3,2)$ FR code based on $K_{3,3}$ and its file size for $1\leq k\leq 3$ are shown in Fig.~\ref{fig:bipartite}.
\begin{figure*}[t]
 \centering
 \includegraphics[width=0.6\textwidth]{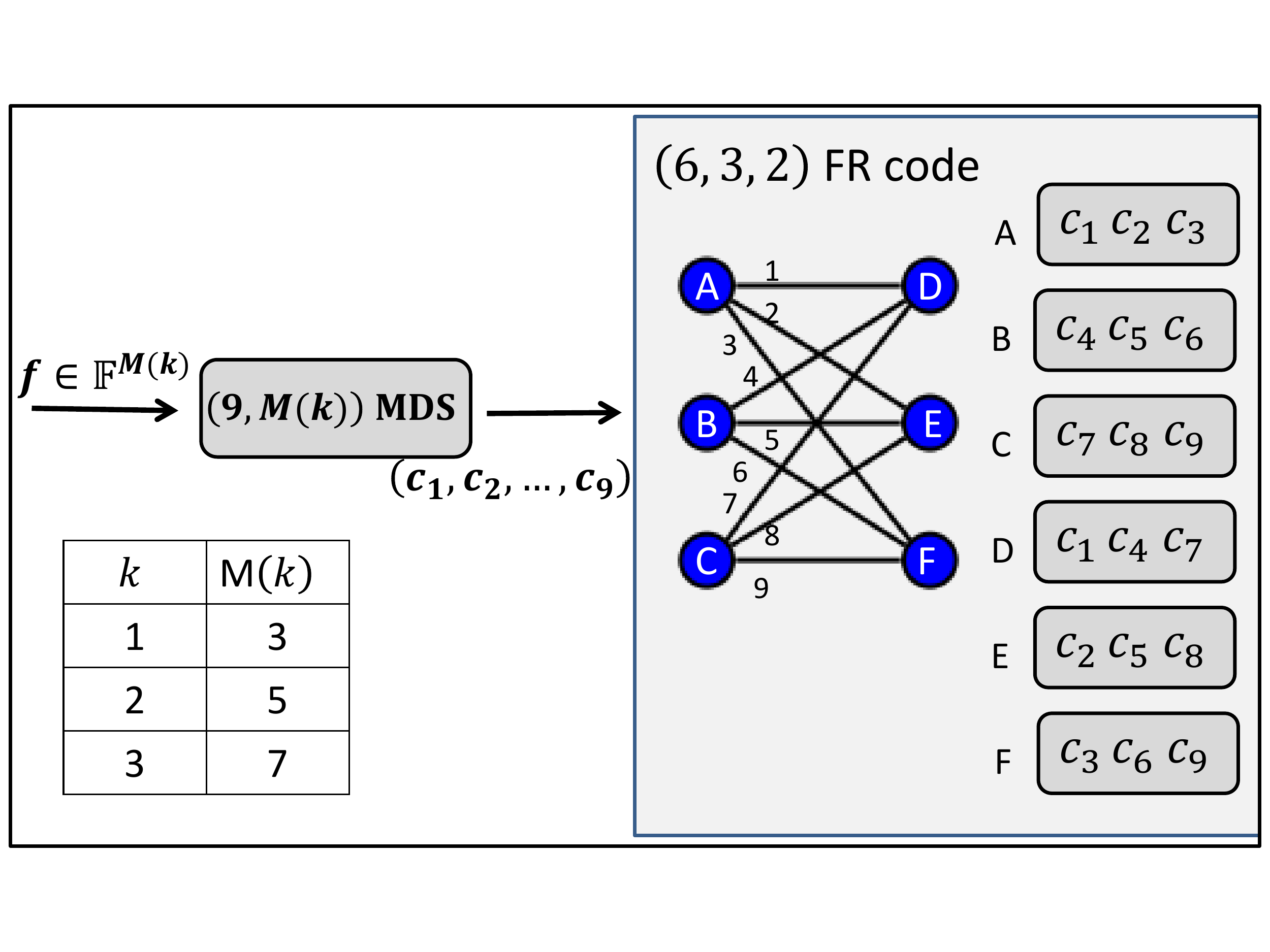}
 \caption{The $((9,M(k)),k,(6,3,2))$ DRESS code with the inner  FR code based on the complete bipartite graph $K_{3,3}$}\label{fig:bipartite}
\end{figure*}
\end{example}


\subsection{$k$-Optimal FR Codes Based on Graphs with a Given Girth}
\label{subsec:k-optimal}

 First, we provide a simple upper bound on the file size of FR codes and show that this bound can be attained. The proof of the following Lemma can be easily verified from~(\ref{eq:bound2}) or Lemma~\ref{lm:isoperimetric}.

\begin{lemma}
\label{lm:upper} If $C$ is an $(n,\alpha, 2)$ FR code then the file size $M(k)$ of $C$, for any $1\leq k\leq \alpha$, satisfies
$$M(k)\leq k\alpha-k+1.
$$
\end{lemma}

By Lemma~\ref{lm:isoperimetric},  to obtain a large value for $M(k)$, every induced subgraph with $k$ vertices should be as sparse as possible. Hence, for the rest of this subsection
we consider graphs with a large girth (usually larger than $k$), in other words, the induced subgraphs with $k$ vertices, $1\leq k\leq \alpha$, will be trees.
Next, we consider the girth of a graph and show that FR codes obtained from a graph with a large enough girth are optimal.

\begin{lemma}
\label{lm:girth}
 Let $G$ be an $\alpha$-regular graph with $n$ vertices and let $M(k)$ be  the file size of the corresponding FR code~$C_{G}$.
The girth of  $G$ is at least $k+1$ if and only if $M(k)=k\alpha -(k-1)$.

\end{lemma}
\begin{proof}
Let $G$ be a graph with girth $g$.
 Any induced subgraph $G'$ of $G$ with $k$ vertices has at most $k-1$ edges if and only if $g\geq k+1$. Clearly, there exists at least one induced subgraph $G'$ of $G$ with $k$ vertices and $k-1$ edges. Thus, by Lemma~\ref{lm:isoperimetric} we have
$M(k)=k\alpha-(k-1).
$
\end{proof}

\begin{corollary}
\label{cor:girthOptimal}
For each $k\leq g-1$,
an FR code $C_G$ based on an $\alpha$-regular graph $G$ with girth $g$ attains the bound in~(\ref{eq:bound2}), and hence it is $k$-optimal. $C_G$ also attains the bound of Lemma~\ref{lm:upper}.
\end{corollary}
\begin{corollary}
\label{cor:Optimal}
An FR code $C_G$ based on an $\alpha$-regular graph $G$ with girth $g\geq \alpha+1$ is optimal.
 \end{corollary}

\begin{theorem}
\label{thm:rateGirth}
Let $G$ be a graph with girth $g$. Then the file size $M(k)$ of an FR code $C_G$ based on $G$ satisfies
\[M(k)=\left\{\begin{array}{cc}
         k\alpha -k+1 \;&\textmd{ if } k\leq g-1 \\
           k \alpha-k \;&\textmd{ if }g\leq k\leq g+\lceil\frac{g}{2}\rceil-2.
         \end{array}\right.
\]
\end{theorem}

\begin{proof}
For $k\leq g-1$ the result follows directly from Lemma~\ref{lm:girth}. Since the graph $G$ has a cycle of length $g$, it follows that $M(g)=g\alpha-g$.
If any subgraph of $G$ with $k$ vertices contains at most one cycle then the file size  satisfies $M(k)=k\alpha-k$.
Note that for $g\leq k\leq g+\lceil\frac{g}{2}\rceil-2$  there is no subgraph of $G$ with $k$ vertices that contains two cycles with no common vertices.
Now we claim that the minimum number $m$ of vertices in a connected subgraph of $G$ with two cycles is $g+\lceil\frac{g}{2}\rceil-1$. Assume for the contrary that $m=g+\lceil\frac{g}{2}\rceil-1-\epsilon$, $\epsilon\geq 1$. Hence there exists a subgraph of $G$ depicted in Fig.~\ref{fig:twoCycles} such that
\begin{figure*}[h]
 \centering
 \includegraphics[width=0.4\textwidth]{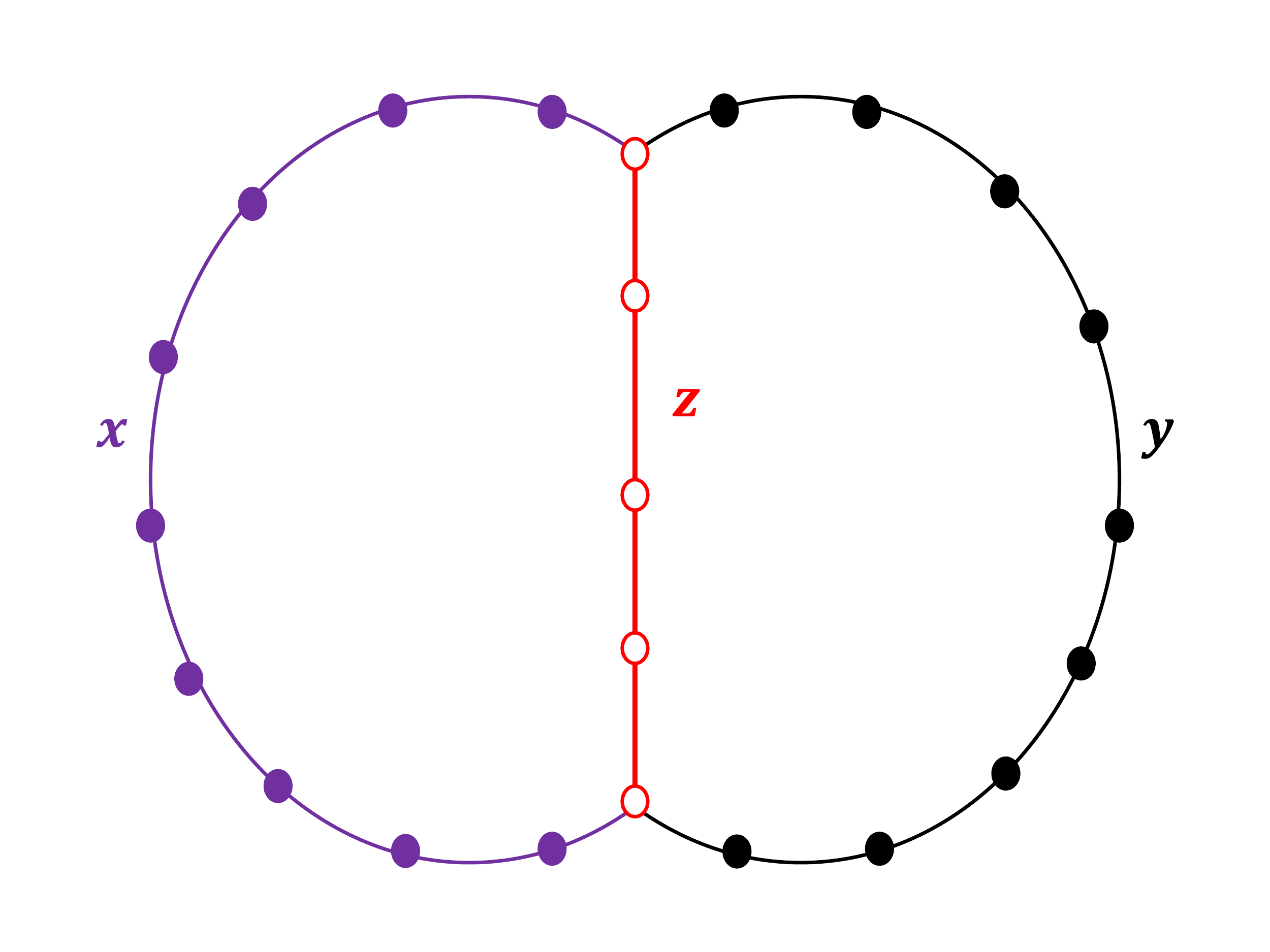}
 \caption{Two intersecting cycles }\label{fig:twoCycles}
\end{figure*}

\begin{equation}
\label{eq:1}
x+z\geq g
\end{equation}
\begin{equation}
y+z\geq g
\label{eq:2}
\end{equation}
\begin{equation}
\label{eq:3}
x+y+2\geq g
\end{equation}
\begin{equation}
\label{eq:4}
x+y+z=g+\lceil\frac{g}{2}\rceil-1-\epsilon.
\end{equation}
Hence from~(\ref{eq:1}) and (\ref{eq:4}) we have that $y\leq \lceil\frac{g}{2}\rceil-2$ and from~(\ref{eq:2}) and~(\ref{eq:4}) we have that $x\leq \lceil\frac{g}{2}\rceil-2$ which contradicts to (\ref{eq:3}). Thus,
the maximum number of vertices in a subgraph of $G$ with at most one cycle is $g+\lceil\frac{g}{2}\rceil-2$ which completes the proof of the theorem.
\end{proof}

Next we consider examples of  FR codes based on some interesting graphs for which the girth is known.

\begin{example}
\label{ex:MooreFR}

Let TD be a $\textmd{TD}(q,q)$ transversal design. The  $(2q^2,q,2)$ FR code $C_{G_{\textmd{TD}}}$ based on the incidence graph $G_I(\textmd{TD})$  attains the bound in~(\ref{eq:bound2}) for all $k\leq 5$.

The following graphs attain the Moore bound (see Theorem~\ref{MooreBound}), i.e., they have the minimum number of vertices given girth and  degree.
%
The parameters of the FR codes corresponding to these graphs can be found in the following table.
\vspace{.3cm}
\begin{center}
\begin{tabular}{|c|c|c|c|}
  \hline
  name of a graph& degree & girth & $(n,\alpha, \rho)$ \\\hline\hline
  Complete graph $K_n$ & $n-1$ & 3 &$(n,n-1,2)$ \\\hline
   Complete bipartite graph $K_{r,r}$ & $r$ & 4 & $(2r,r,2)$ \\\hline
  Petersen graph & 3 & 5 & $(10,3,2)$ \\\hline
  Hoffman-Singleton graph & 7 & 5 & $(50,7,2)$ \\\hline
  Projective plane & $q+1$ & 6 & $(2q^2+2q+2,q+1,2)$ \\\hline
  Generalized quadrangle & $q+1$ & 8 & $(2q^3+2q^2+2q+2,q+1,2)$ \\\hline
  Generalized hexagon & $q+1$ & 12 & $(2q^5+2q^4+2q^3+2q^2+2q+2,q+1,2)$ \\
  \hline
\end{tabular}
\end{center}

\end{example}

\vspace{.3cm}

Next, we use the Moore bound to show that the bound in~(\ref{eq:bound2}) can be improved in some cases.

\begin{lemma}
\label{lm:not tight}
The  bound in~(\ref{eq:bound2}) is not tight  for $\rho=2$ if $$\alpha k-\alpha-k+3\leq n< \widehat{n}_0(\alpha,k+1),$$
 where
$ \widehat{n}_0(d,g)$  is the minimum number of vertices in a $(d,g)$-cage (see Section~\ref{sec:preliminaries}).
\end{lemma}

\begin{proof} Let $n,k$ be  integers such that $\alpha k-\alpha-k+3\leq n< \widehat{n}_0(\alpha,k+1)$. Since $n< \widehat{n}_0(\alpha,k+1)$, the $\alpha$-regular graph with $n$ vertices  corresponding to an $(n,\alpha,2)$ FR code has girth at most $k$, and hence by Lemma~\ref{lm:girth} the file size satisfies $M(k)\leq k\alpha-k$. Therefore, $A(n,k,\alpha,2)\leq k\alpha-k$.

To complete the proof, we will show that if $\alpha k-\alpha-k+3\leq n$ then $\varphi(k)=k\alpha-k+1$. To prove this statement we will prove by induction that if $\alpha k-\alpha-k+3\leq n$ then for all $1\leq\ell\leq k$ it holds that $\varphi(\ell)=\ell\alpha-\ell+1$.  This trivially holds for $\ell=1$. Assume that $\varphi(\ell-1)=(\ell-1)\alpha-(\ell-1)+1$, $2\leq \ell\leq k$. Then, by applying the recursion in~(\ref{eq:bound2}) we have
$$\varphi(\ell)=(\ell-1)\alpha-(\ell-1)+1+\alpha-\left\lceil\frac{(\ell-1)\alpha-2\ell+4}{n-\ell+1}\right\rceil=\ell\alpha-\ell+1,
$$
since $(\ell-1)\alpha-2\ell+4\leq n-\ell+1$ for $n\geq\alpha k-\alpha-k+3$ and $\ell\leq k$.

\end{proof}

As a consequence of Lemma~\ref{lm:not tight} we have that the bound in~(\ref{eq:bound2}) is not always tight and hence we have a similar better bound on $A(n,k,\alpha, \rho)$:
 \begin{equation*}
   A(n,k,\alpha, \rho)\leq  \varphi'(k), \textmd{ where } \varphi'(1)=\alpha,
  \end{equation*}
  \[\textmd{  } \varphi'(k+1)= A(n,k,\alpha, \rho)+\alpha-\left\lceil\frac{\rho  A(n,k,\alpha, \rho)-k\alpha}{n-k}\right\rceil.\]

\subsection{FR Codes with a Given File Size}
\label{subsec:file_size}

We observe from Lemma~\ref{lm:clique} and Lemma~\ref{lm:upper}
that for any $1\leq k\leq \alpha$, the file size $M(k)$ of an $(n,\alpha,2)$  FR code $C$ satisfies
\begin{equation}
\label{eq:range}
k\alpha-\binom{k}{2}\leq M(k)\leq k\alpha-(k-1),
\end{equation}
and the value of the file size depends on the structure of the underlying $\alpha$-regular graph $G$.  If the graph contains a clique~$K_k$ then the file size attains the lower bound in~(\ref{eq:range}). If the graph does not contain a cycle of length $k$ then the file size attains the upper bound in~(\ref{eq:range}). The intermediate values for the file size can be obtained by excluding  certain subgraphs of $K_k$ from the graph $G$. For example, $M(k)\geq k\alpha-\binom{k}{2}+1$ if and only if $G$ does not contain $K_k$ as a subgraph, and
to have $M(k)\geq k\alpha-\binom{k}{2}+2$,  $G$ should not contain $K_k-e$, i.e., a $k$-clique without an edge.
Note that the problem of finding the minimum number of vertices in a graph which does not contain a specific subgraph is highly related to a Tur\'an type problems (see e.g.,\cite{FuSi13} and the references therein).

For the rest of the section we assume that $n\alpha$ is an even integer.

There are only two possible values for $M(3)$, $3\alpha-3$ and $3\alpha-2$. To have a code $C$ with  file size $3\alpha-2$, one should exclude a clique $K_3$, which is  also a cycle of length $3$.
From~(\ref{eq:cage}) it follows that  for a given $\alpha$, if $n<2\alpha$ then $M(3)=3\alpha-3$. Equivalently, the necessary condition for $M(3)=3\alpha-2$ is that $n\geq2\alpha$.

Constructions for codes  with file size $M(3)=3\alpha-2$ are provided in the previous subsection, based on optimal (Tur\'an, Moore) graphs, for  specific choices for the parameters $\alpha, n$, where  $n$ is even. In addition, we provide in Appendix~\ref{app:triangle} another two constructions of FR codes with file size $3\alpha-2$, the first one for even $n\geq 2\alpha$, and the second one for odd $n\geq \frac{5}{2}\alpha$.

The following lemma is proved in Appendix~\ref{app:triangle}.
\begin{lemma}
\label{lm:odd-triangle}
Let $n_3$ is the minimum value of $n$ such that $A(n,3,\alpha, 2)=3\alpha-2$, for any $n\geq n_3$. Then
\[2\alpha+2\leq n_3\leq \frac{5\alpha}{2}.
\]
\end{lemma}

We conjecture that $n_3=\frac{5}{2}\alpha$.
The following theorem is an immediate consequence from the discussion above (see also Appendix~\ref{app:triangle}).

\begin{theorem}
\label{thm:summary3}
The maximum file size for $k=3$ satisfies
\begin{itemize}
  \item For even $n$
  \[A(n,3,\alpha,2)=
  \left\{\begin{array}{cc}
  3\alpha-3&\textmd{ if } \;n < 2\alpha \\
           3\alpha-2&\textmd{ if }\; n\geq 2\alpha
         \end{array}\right..
  \]
  \item For odd $n$ and even $\alpha$, $A(n,3,\alpha,2)=3\alpha-2$.
   \[A(n,3,\alpha,2)=
  \left\{\begin{array}{cc}
  3\alpha-3&\textmd{ if } \;n < n_3 \\
           3\alpha-2&\textmd{ if }\; n\geq n_3
         \end{array}\right..
  \]

\end{itemize}

\end{theorem}

\vspace{.5cm}

By (\ref{eq:range}) we have that $M(4)\in\{4\alpha-3, 4\alpha-4, 4\alpha-5, 4\alpha-6\}$.

\begin{enumerate}
\item If $M(4)=4\alpha-3$ then by Lemma~\ref{lm:girth} the corresponding graph $G$ has girth at least 5.  Codes with file size $4\alpha-3$  can be derived from any graph $G$ with girth $\geq 5$, e.g., Hoffman-Singleton graph and its {generalizations~\cite{Mu79,AABL12}.}
\item If $M(4)=4\alpha-4$ then the corresponding graph $G$ contains a subgraph of $K_4$ with $4$ edges,
but does not contain $K_4-e$, a $4$-clique without an edge. Codes with file size $4\alpha-4$ and minimum number of nodes are constructed from $K_{\alpha,\alpha}$ by Theorem~\ref{trm:Turan}.
\item If $M(4)=4\alpha-5$ then the corresponding graph $G$ contains $K_4-e$,
but does not contain $K_4$. Codes with file size $4\alpha-5$ and
minimum number of nodes are constructed from $(n,3)$-Tur\'an graphs if $n\equiv 0\;(\textmd{mod }3)$. If $n\not\equiv 0\;(\textmd{mod }3)$ then  we use a modification of a Tur\'an graph (see Example~\ref{ex:ap-motdification} in Appendix~\ref{app:triangle}).
\item If $M(4)=4\alpha-6$ then the corresponding graph $G$ contains $K_4$.
Codes with file size $4\alpha-6$ and  minimum number of nodes are constructed from the complete graph $K_{\alpha+1}$.
\end{enumerate}

Let  $G(k,\ell)$ denote any graph with $k$ vertices and $\ell$ edges.
By Lemma~\ref{lm:isoperimetric}, to calculate the file size of an FR code based on a graph $G$, we need to find an induced subgraph $G(k,\ell)$ of $G$ with the largest $\ell$. The following lemma is an immediate consequence from Lemma~\ref{lm:isoperimetric}.

\begin{lemma}
\label{lm:G-M}
 Let $G$ be an $\alpha$-regular graph and let $C_G$ be the FR code based on $G$.
\begin{itemize}
  \item If $G$ contains an induced subgraph $G(k,\ell)$ then $M(k)\leq k\alpha-\ell$.
  \item If $G$ does not contain an induced subgraph $G(k,r)$  for all $r\geq \ell+1$ then $M(k)\geq k\alpha-\ell$.
  \item If $G$ contains an induced subgraph $G(k,\ell)$  and does not contain  $G(k,r)$ for $r\geq \ell+1$ then $M(k)=k\alpha-\ell$.
\end{itemize}
\end{lemma}

Based on the previous discussion  we have the following theorem.

\vspace{.5cm}

\begin{theorem}
\label{thm:k4} The maximum file size for $k=4$ satisfies

  \[A(n,4,\alpha,2)\leq
  \left\{\begin{array}{cl}
  4\alpha-6&\textmd{ if } \;\alpha+1\leq n < \frac{3}{2}\alpha \\
           4\alpha-5&\textmd{ if }\; \frac{3}{2}\alpha\leq n<\min\{n_3, 3\alpha-3\}\\
           4\alpha-4&\textmd{ if }\; \min\{n_3, 3\alpha-3\}\leq n < 1+\alpha^2\\
           4\alpha-3&\textmd{ if }\; n\geq 1+\alpha^2
         \end{array}\right.,
  \]

  \end{theorem}
\begin{proof}
The range of $n$ for file size $4\alpha-6$ follows from Corollary~\ref{cor:Turan}. For file size $4\alpha-3$, the range of $n$ follows from the Moore bound (see Theorem~\ref{MooreBound}).

To distinguish between file sizes $4\alpha-4$ and $4\alpha-5$, we note that for $4\alpha-4$ the corresponding graph $G$ should not contain an induced subgraph  $G(4,5)$. If $G$ does not contain $K_3$ then obviously it does not contain $G(4,5)$. Hence, by Theorem~\ref{thm:summary3}, $n\geq n_3$ for file size $4\alpha-4$. If $G$ contains $K_3$ then let $v, u_1$ and $u_2$ be three vertices of~$K_3$. Since the degree of $v$ is $\alpha$, it follows that there are $\alpha-2$ other adjacent vertices $u_3,\ldots,u_{\alpha}$ of $v$. Since there is no $G(4,5)$ in $G$ it follows that $u_1$ and $u_2$ are not adjacent to any vertex in $\{u_3,\ldots,u_{\alpha}\}$. Moreover, except for $v$ there is no other vertex for which both $u_1$ and $u_2$ are adjacent. Let $A$ ($B$, respectively) be the set of additional $\alpha-2$ vertices to which $u_1$ ($u_2$, respectively) is adjacent. The set $A\cup B\cup \{v\}\cup\{u_1,u_2,u_3,\ldots,u_{\alpha}\} $ contains $3\alpha-3$ vertices, and hence $n\geq \min\{n_3, 3\alpha-3\}$, which completes the proof of the theorem.
\end{proof}

\begin{remark}Note, that there is an inequality in Theorem~\ref{thm:k4} since it is not always possible to find
a regular graph which satisfies a given constraint on the number of edges in a subgraph with a given number of vertices. For example, there are no Moore graphs  for some parameters.
\end{remark}

The existence problem of FR codes
 with $\rho=2$, for any given file size  in the interval between $k\alpha -\binom{k}{2}$ and $k\alpha-k+1$, is getting more complicated as
$k$ increases. Some  file sizes can be obtained by graphs with a given girth and by $(n,r)$-Tur\'an graphs with different $r$'s (see Theorem~\ref{trm:Turan} and Theorem~\ref{thm:rateGirth}). However, not all the values between $k\alpha -\binom{k}{2}$ and $k\alpha-k+1$ can be obtained by this way.
Moreover, an additional problem is to find the minimum number of vertices in an $\alpha$-regular  graph $G$ such that  the corresponding FR code $C_G$ has file size $k\alpha-x$, for some $k-1\leq x\leq \binom{k}{2}$. Note, that in some cases  Theorem~\ref{trm:Turan} and Theorem~\ref{thm:rateGirth} together with Moore bound and Tur\'an's theorem can provide the answer to this problem, as was demonstrated in Theorem~\ref{thm:summary3} and Theorem~\ref{thm:k4}.

The following lemma shows that given an information about the file size for a given reconstruction degree $k$, one can get some information about the file size for the reconstruction degree $k+1$. In other words, by eliminating the existence of some induced subgraphs of size $k$, one can prove the nonexistence of induced subgraphs of size $k+1$ which implies bounds on $M(k+1)$ by Lemma~\ref{lm:isoperimetric}.

\begin{lemma}
\label{lm:G-recursion}
If $G$ is a regular graph that does not contain $G(k,\binom{k}{2}-\delta)$, for all $0\leq \delta\leq x$, where $x$ is a given integer,  $0\leq x\leq \binom{k-1}{2}$,
then $G$ also does not contain $G(k+1,\binom{k+1}{2}-\gamma)$,  for all $0\leq \gamma\leq y$, where
 $y$ is given by
\begin{equation}
\label{eq:G(k,l)}
y=\left\{\begin{array}{cl}
    \binom{k}{2}-1 &\textmd{ if } x =\binom{k-1}{2} \\
  \left\lfloor\frac{2x}{k-1}\right\rfloor+x+1 &\textmd{ if } k \textmd{ is odd and } x\leq \binom{k-1}{2}-1\\
           2\left\lfloor\frac{x}{k-1}\right\rfloor+x+1 &\textmd{ if } k \textmd{ is even, } x\leq \binom{k-1}{2}-1\textmd{ and } 0\leq x\textmd{ mod }k-1\leq \frac{k-4}{2}\\
           2\left\lfloor\frac{x}{k-1}\right\rfloor+x+2 &\textmd{ if } k \textmd{ is even, } x\leq \binom{k-1}{2}-1\textmd{ and } \frac{k-2}{2}\leq x\textmd{ mod }k-1\leq k-2\\
         \end{array}\right.
\end{equation}

\end{lemma}

\begin{proof}
First, we prove the statement for the maximum value  of $\gamma$.
If there exists an induced subgraph $G(k+1,\binom{k+1}{2}-y)$ in $G$ then in its complement $\overline{G}(k+1,\binom{k+1}{2}-y)$ there are $y$ edges and hence in $\overline{G}(k+1,\binom{k+1}{2}-y)$ there is a vertex $v$ with degree at least $\left\lceil\frac{2y}{k+1}\right\rceil$.
However, $G(k+1,\binom{k+1}{2}-y)\setminus \{v\}$ is an induced subgraph $G(k, \binom{k}{2}-(y-\left\lceil\frac{2y}{k+1}\right\rceil)+\epsilon)$, for some $\epsilon\geq 0$. To prove the statement of the lemma, it is sufficient to show that \begin{equation}
\label{eq:proofG}
x=y-\left\lceil\frac{2y}{k+1}\right\rceil,
\end{equation}
for $y$ given in~(\ref{eq:G(k,l)}).
It is easy to verify that $y$ from~(\ref{eq:G(k,l)}) indeed satisfies equation~(\ref{eq:proofG}).
The proof for the other values of $\gamma$ follows by induction.
\end{proof}

\begin{corollary}
If $M(k)\leq k\alpha-(\binom{k}{2}-x-1)$, for a given $x$, $0\leq x\leq \binom{k-1}{2}$, then $M(k+1)\leq (k+1)\alpha-(\binom{k+1}{2}-y-1)$, where $y$ is given in ~(\ref{eq:G(k,l)}).
\end{corollary}

As we already saw, constructions of FR codes for given specific parameters can be formulated in terms of graph theory.
We formulate the following problems in graph theory that provide some information about the number of nodes of FR codes with maximum file size and
the existence of an FR code with a given file size.

\textbf{Problem 1.} Find the value of  $N(k, \alpha, \delta)$, $0\leq\delta\leq \binom{k}{2}-k$, which is the maximum number of vertices such that any $\alpha$-regular graph $G$  with $n<N(k, \alpha, \delta)$ vertices, where $n\alpha$ is even integer,  contains $G(k,\binom{k}{2}-x)$, for some $0\leq x\leq \delta$.

\textbf{Problem 2.} Find the value of  $N'(k, \alpha, \delta)$, $0\leq\delta\leq \binom{k}{2}-k$,  which is the minimum number of vertices such that for any $n\geq N'(k, \alpha, \delta)$, where $n\alpha$ is even integer,  there exists an
$\alpha$-regular graph $G$  with $n$ vertices which does not contain $G(k,\binom{k}{2}-x)$, for any $0\leq x\leq \delta$ but contains $G(k,\binom{k}{2})-\delta-1)$.

\textbf{Problem 3.} Let $n, k, \alpha, \delta$ be positive integers such that $3\leq k\leq \alpha$ and $0\leq \delta\leq \binom{k}{2}-k$. Does there exist an $\alpha$-regular graph $G$ with $n$ vertices which does not contain $G(k,\binom{k}{2}-x)$ for any $0\leq x\leq \delta$ but contains $G(k, \binom{k}{2}-\delta-1)$?

Clearly, an answer to Problem 3 provides a solution to the existence question of  FR codes with any file size for $\rho=2$. Based on Lemma~\ref{lm:G-M} and the solutions to Problem 1 and Problem 2, one can prove the following lemma.
\begin{lemma}
\[
\begin{array}{cc}
    A(n,k,\alpha,2) \leq k\alpha-\binom{k}{2}+\delta &\textmd{ if }n< N(k, \alpha, \delta)\\
     A(n,k,\alpha,2)\geq  k\alpha-\binom{k}{2}+\delta+1 &\textmd{ if } n\geq N'(k, \alpha, \delta).
  \end{array}
\]
\end{lemma}

By our discussion above we have
\begin{corollary}
 $N(3,\alpha,0)=2\alpha$; $N'(3,\alpha,0)=2\alpha $ if $\alpha$ is odd, and $ N'(3,\alpha,0)\leq \frac{5}{2}\alpha$ if $\alpha$ is even.
 \end{corollary}

 From the Moore bound and from the Tur\'an bound we have the following corollary.

\vspace{.2cm}

 \begin{corollary}
$~$
 \begin{itemize}
   \item $N(k,\alpha,\binom{k}{2}-k)\geq \widehat{n}_0(\alpha, k+1)$, where $\widehat{n}_0(\alpha, k)$ is defined in Subsection~\ref{subsec:graphs}.
   \item $N(k,\alpha,0)\geq\left\lceil \frac{k-1}{k-2}\alpha\right\rceil$.
 \end{itemize}

 \end{corollary}

\textbf{Problem 4.}  Find the value of  $\eta(k, \alpha, \delta)$, $k-1\leq\delta\leq \binom{k}{2}$,  which is the minimum number of vertices in  a  graph  such that the FR code based on $G$ has file size $k\alpha-\delta$.

\vspace{.5cm}

Now we consider bounds on $\eta(5, \alpha, \delta)$.

\begin{itemize}

\item A code $C$ for which $M(5)= 5 \alpha -4$ is obtained from  a graph $G$ if and only if the girth of $G$ is at least 6. Hence, $\eta(5, \alpha, 4)\geq 2\alpha^2-2\alpha+2$, by the Moore bound (if $\alpha=q+1$, for a prime power $q$, then the existence of a projective plane of order $q$ implies that $\eta(5, \alpha, 4)= 2\alpha^2-2\alpha+2$).

\item A code $C$ for which $M(5)= 5 \alpha -5$ can be obtained from any graph with girth 5. Hence, if there exists an $(\alpha, 5)$-cage then $\eta(5, \alpha, 5)\leq \alpha^2+1$, by the Moore bound.

\item A code $C$ for which $M(5)= 5 \alpha -6$ can be obtained from an $(n,2)$-Tur\'an graph, which is also a graph with girth~4. Hence,  $\eta(5, \alpha, 6)\leq 2\alpha$, by the Tur\'an's theorem.

\item A code $C$ for which $M(5)= 5 \alpha -8$ can be obtained from an $(n,3)$-Tur\'an graphs. Hence, if   $\alpha$ is even  then $\eta(5, \alpha, 8)\leq \frac{3}{2}\alpha$, by the Tur\'an's theorem.

\item A code $C$ for which $M(5)= 5 \alpha -9$ can be obtained from an $(n,4)$-Tur\'an graph. Hence, if  3 divides $\alpha$ then $\eta(5, \alpha, 9)\leq \frac{4}{3}\alpha$, by the Tur\'an's theorem.

\item A code $C$ for which $M(5)= 5 \alpha -10$ can be obtained  from an $(n,r)$-Tur\'an graph, for any $r\geq 5$ (e.g., a complete graph for $r=n$). Hence, $\eta(5, \alpha, 10)= \alpha+1$.

\end{itemize}

One can see that the value of $M(5)=5\alpha-7$ is missing from the list above. However, it is easy to prove that any file size in the interval between $k\alpha -\binom{k}{2}$ and $k\alpha-k+1$ can be obtained.
To prove this claim one can start with
a known graph $G=(V,E)$ which contains an induced subgraphs $G(k,t-1)$ but does not contain $G(k,t)$.
Let $H=G(k,t-1)$ be an induced subgraph
of $G$; let $v_1,v_2$ two vertices
in $H$ and $v_3, v_4$ two vertices in $G\setminus H$ such
that $e_1=\{v_1,v_2\}, e_2=\{v_3,v_4\}\notin E$ and $e_3=\{v_1,v_3\}, e_4=\{v_2,v_4\}\in E$.
One can easily verify that the graph $G\setminus \{e_3,e_4\}\cup\{e_1,e_2\}$
contains an induced subgraph $G(k,t)$ but does not contain $G(k,t+1)$.
This implies  that $\eta(k, \alpha, \delta)$ is a decreasing function of $\delta$, and hence $\eta(5,\alpha, 7)\leq \eta(5,\alpha, 6)$.
We conjecture that $\eta(5,\alpha, 7)=\ell\alpha$, where $\frac{3}{2}<\ell<2$.
%

The discussion of this subsection about FR codes with a given file size raises a lot of questions in graph theory. We leave the research on these questions for future research in graph theory.


\section{Fractional Repetition Codes with Repetition Degree $\rho>2$}
\label{sec:rho>2}

In this section, we consider FR codes with repetition degree $\rho> 2$. Note, that while  codes with $\rho=2$ have the maximum data/storage ratio,  codes with $\rho>2$ provide multiple choices for node repairs. In other words, when a node fails, it can be repaired from different $d$-subsets of available nodes.

We present  generalizations  of the  constructions from the previous section which were based on Tur\'an graphs and graphs with a given girth. These generalizations employ transversal designs and generalized polygons, respectively.

We start with a construction of FR codes from transversal designs.  Let $\textmd{TD}$ be a transversal design $\textmd{TD}(\rho,\alpha)$, $\rho\leq \alpha+1$, with block size $\rho$ and group size $\alpha$. Let $C_{\textmd{TD}}$ be an $(n,\alpha,\rho)$ FR code based on $\textmd{TD}$ (see Subsection~\ref{subsec:codes}). Recall that
 by Lemma~\ref{lm:TDparameters}, there are $\rho \alpha$ points in $\textmd{TD}$ and hence  $n= \rho \alpha$. Note, that  all the symbols stored in  node $i$ correspond to the set $N_i$ of blocks from $\textmd{TD}$ that contain the point $i$. Since by Lemma~\ref{lm:TDparameters} there are $\alpha$  blocks that contain a given point, it follows that each node stores $\alpha$ symbols.

Similarly to  Theorem~\ref{trm:Turan} we can prove the following theorem.

\begin{theorem}
\label{lm:TDrate}
Let $k=b\rho+t$, for $b,t\geq 0$ such that $t\leq \rho-1$. For an $(n=\rho\alpha,\alpha,\rho)$ FR code $C_{\textmd{TD}}$  based on a transversal design $\textmd{TD}(\rho,\alpha)$ we have

\[M(k)\geq k\alpha-\binom{k}{2}+\rho\binom{b}{2}
+bt.
 \]
\end{theorem}

\begin{remark} Note, that for all $k\geq \rho+1$, the file size of the FR code $C_{\textmd{TD}}$ is strictly larger than the MBR capacity.
\end{remark}

\begin{corollary}
Let $C_{\textmd{TD}}$ be an $(r\alpha,\alpha,r)$ FR code based on $\textmd{TD}$, a transversal design $\textmd{TD}(r,\alpha)$. Let $C_T$ be an $(\frac{r}{r-1}\alpha,\alpha,2)$ FR code based on $(n,r)$-Tur\'an graph $T$. If $M_{\textmd{TD}}(k)$ and $M_{T}(k)$ are their file sizes, respectively, then
\begin{enumerate}
  \item $C_{\textmd{TD}}=C_T$ for $r=2$;
  \item $M_{\textmd{TD}}(k)\geq M_{T}(k)$ for all $r\geq 2$.
\end{enumerate}
\end{corollary}

\begin{example}
\label{ex:TD_34}
Let TD be a transversal design $\textmd{TD}(3,4)$  defined as follows:
 $\cP=\{1,2,\ldots,12\}$; $\mathcal{G}=\{G_1,G_2,G_3\}$, where
$G_1=\{1,2,3,4\}$, $G_2=\{5,6,7,8\}$, and $G_3=\{9,10,11,12\}$;
$\mathcal{B}=\{B_1,B_2,\ldots, B_{16}\}$,
and incidence matrix  given by

\begin{footnotesize}
\[\textbf{I}(\textmd{TD})=\left(
                   \begin{array}{cccccccccccccccc}
                     1 & 1 & 1 & 1 & &&&&&&&&&&& \\
                     &&& & 1 & 1 & 1 & 1 & &&&&&&& \\
                     &&&&&&&& 1 & 1 & 1 & 1 & &&& \\
                     &&&&&&&&&&& & 1 & 1 & 1 & 1 \\
                     1 & && & 1 & && & 1 & && & 1 & && \\
                      & 1 & &&& 1 & && & 1 & && & 1 & & \\
                      &  & 1 & && & 1 & && & 1 & && & 1 &  \\
                      & & & 1 & && & 1 & && & 1 & && & 1 \\
                     1 & &&& & 1 & &&&& & 1 & & & 1 &  \\
                      & 1 & & & 1 & &&&& & 1 & &&& & 1 \\
                      &  & 1 & &&& & 1 &  & 1 & & & 1 & && \\
                      & & & 1 & & & 1 &  & 1 & &&& & 1 & & \\
                   \end{array}
                 \right)
\]
\end{footnotesize}

The placement of the symbols from a codeword of the corresponding MDS code of length $16$ is shown in Fig.\ref{fig:TD34}.
\begin{figure*}[t]
 \centering
 \includegraphics[width=0.5\textwidth]{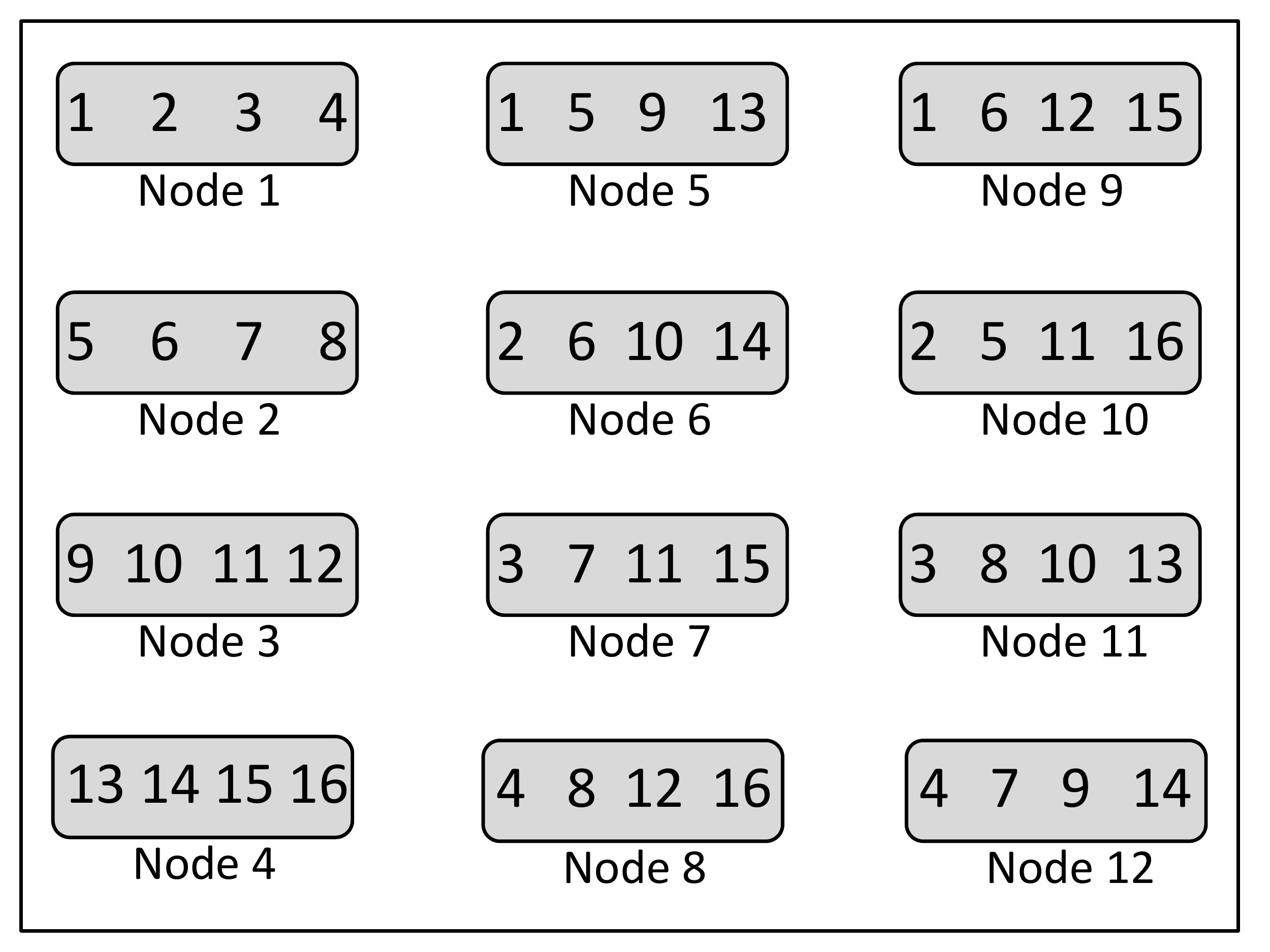}
 \caption{The FR code based on \textmd{TD}(3,4)}\label{fig:TD34}
\end{figure*}
The  values of the file size $M(k)$ for $1\leq k\leq 4$ are given in the following table.
\[\begin{tabular}{|c|c|}
    \hline
    $k$ & $M(k)$ \\  \hline\hline
    1 & 4 \\\hline
    2 & 7 \\\hline
    3 & 9 \\\hline
    4 & 11 \\
    \hline
  \end{tabular}
\]

\end{example}

In the following theorem, proved in Appendix~\ref{app:thm:alphaBound}, we find the conditions on the parameters such that the bound on the file size of an FR code $C_{\textmd{TD}}$ from Theorem~\ref{lm:TDrate} attains the recursive bound in~(\ref{eq:bound2}).
\begin{theorem}
\label{thm:alphaBound}

Let $\rho\geq 3$, $k=b\rho+t\leq \alpha$,  $0\leq t\leq \rho-1$, and $\alpha > \alpha_0(k)$, where
\begin{equation*}
\alpha_0(k)=\left\{\begin{array}{cc}
         \frac{ b^2\rho\binom{\rho-1}{2}+(\rho-2)\binom{t-1}{2}+b((\rho^2+1)(t-1)-\rho(3t-4))}{\rho-t+1}& \;\textmd{ if }k\not\equiv 0\;(\textmd{mod }\rho) \\
           b(b\rho-2)\binom{\rho-1}{2}+b-1&\;\textmd{ if } k\equiv 0\;(\textmd{mod }\rho)
         \end{array}\right..
\end{equation*}
The file size $M(k)$ of the $(\rho\alpha,\alpha,\rho)$ FR code $C_{\textmd{TD}}$ is given by
\[M(k)= k\alpha-\binom{k}{2}+\rho\binom{b}{2}
+bt
 \]
and attains the bound in~(\ref{eq:bound2}) for all $k\leq \alpha$.
\end{theorem}

\begin{example}
We illustrate the minimum values of $\alpha$ for which the FR code obtained from a $\textmd{TD}(\rho, \alpha)$ is optimal as a consequence of Theorem~\ref{thm:alphaBound}.

\begin{center}
\begin{tabular}{|c|c|c|c|c|c|}
  \hline
  \backslashbox{k}{$\rho$} & $3$ & $4$ & $5$ & $6$  \\ \hline
  $3$  & 2 & 3 & 4 & 5 \\\hline
  $4$  & 7 & 7 & 10 & 13 \\\hline
  $5$ & 8 & 17 & 19 & 25 \\\hline
  $6$  & 10 & 22 & 36 & 41 \\\hline
  $7$  & 19& 29 & 47 & 67 \\\hline
  $8$ & 21 & 38 & 61 & 86 \\
  \hline
\end{tabular}
\end{center}
\end{example}
\vspace{.5cm}

Similarly to the case in which $\rho=2$, we continue to find the conditions when there exists an FR code $C$ with file size $M(3)=3\alpha-2$.
To have a file size greater than $3\alpha-3$, we should avoid the existence of a $3\times 3$
 submatrix $\textbf{I}'$ of $\textbf{I}(C)$  such that each row of $\textbf{I}'$ has exactly two ones (recall that the intersection between any two rows is at most one). Such a matrix $\textbf{I}'$ will be called a \emph{triangle}.

\begin{lemma}
\label{lm:k=3general}
If $n < \rho(\rho-1)\alpha-\rho(\rho-2)$ then there exists a triangle in the incidence matrix of an FR code $C$, i.e., a necessary condition for $M(3)=3\alpha-2$ is that $n\geq\rho(\rho-1)\alpha-\rho(\rho-2)$.
\end{lemma}
\begin{proof}
W.l.o.g. assume that the $\alpha$ $1$'s of the first row  of $\textbf{I}(C)$ are in the first $\alpha$ columns.
Let $S$ be the set of $(\rho-1)\alpha$ rows of $\textbf{I}(C)$ which have common ones with the first row of $\textbf{I}(C)$.
Each row of $S$ contains exactly one $1$ in the first $\alpha$ columns  and ${\alpha-1}$ $1$'s in other columns of $\textbf{I}(C)$.
To avoid a triangle in the matrix, the $1$'s in $S$ which do not appear in the first $\alpha$ columns must appear in different columns. Hence, we have
 $\theta\geq \alpha+(\alpha-1)(\rho-1)\alpha=(\rho-1)\alpha^2-(\rho-2)\alpha$. Since $\theta=\frac{n\alpha}{\rho}$ the claim of the lemma is proved.
\end{proof}


By Lemma~\ref{lm:k=3general} it follows that for $n< \rho(\rho-1)\alpha-\rho(\rho-2)$ and $k=3$ the file size of an FR code $C$ equals $M(3)=3\alpha-3$. However, the bound in~(\ref{eq:bound2}) satisfies $\varphi(3)=3\alpha-1-\left\lceil\frac{2(\rho-1)\alpha-\rho}{n-2}\right\rceil=3\alpha-2$ if and only if $n\geq 2(\rho-1)\alpha-(\rho-2)$. Thus, we have the following lemma

\begin{lemma}
The bound in~(\ref{eq:bound2}) for $k=3$ and $\rho>2$ is not tight in
the interval $n\in [2(\rho-1)\alpha-(\rho-2),\rho(\rho-1)\alpha-\rho(\rho-2))$.

\end{lemma}

Next, we present a construction of an FR code $C$ based on a generalized quadrangle,  for which $M(3)=3\alpha-2$. This code attains the bound on $n$ presented in Lemma~\ref{lm:k=3general}.
The following lemma  follows directly from the definition of a generalized quadrangle.

\begin{lemma}
Let $\textmd{GQ}$ be a generalized quadrangle $\textmd{GQ}(s,t)$, where $t\geq s$, and let $C_{\textmd{GQ}}$ be the FR code based on $\textmd{GQ}$.
$C_{\textmd{GQ}}$ is an $(n=(s+1)(st+1),\alpha=t+1,\rho=s+1)$ FR code for which $M(3)=3\alpha-2$ and $M(4)=4\alpha-4$. Moreover, this code attains the bound on $n$ of Lemma~\ref{lm:k=3general}.
\end{lemma}

\begin{remark} Similarly to an FR code $C_{G}$ with $\rho=2$ based on a graph $G$ with girth $g$, we can consider an FR code $C_{\textmd{GP}}$ based on a generalized $g$-gon (generalized polygon GP) for $\rho>2$. One can prove that the file size of $C_{\textmd{GP}}$ is identical to the file size of $C_{G}$ for $k\leq g+\lceil\frac{g}{2}\rceil-2$ given in Theorem~\ref{thm:rateGirth}. However, a generalized $g$-gon  is known to exist only for $g\in\{3,4,6,8\}$.
This observation also holds for a biregular bipartite graph of girth $2g$.  The existence of such graphs was considered in~\cite{AABL12,AABLL13,ABV09,ExJa08,FLSUW95}.
\end{remark}

\begin{remark}
Note that both generalized quadrangles and transversal designs are examples of \emph{partial geometries}. Codes for distributed storage systems based on partial geometries were also considered in~\cite{Hollmann}.
\end{remark}

 Note that the problem of constructing for FR codes with repetition degree greater than two, based on combinatorial designs, also can be considered in terms of \emph{expander} graphs (see e.g~\cite{GUV07}). Let $C_D$ be an FR code based on a combinatorial design $D$. If we consider the incidence graph $G_D=(L\cup R,E)$ of $D$, where $L$ corresponds to the points and $R$ corresponds to the blocks of $D$, then calculating $M(k)$ can be described by calculating the set of neighbours of any subset of size $k$ of the part $L$ of  $G_D$. In other words, for an FR code with the file size $M(k)$ it should hold that $|\Gamma(A)|\geq M(k)$ for every $A\subseteq L$ of size $k$, where $\Gamma(A)$ denotes the set of neighbours of $A$. Hence, to have an FR code with $M(k)$, we need to construct a $(k,\frac{M(k)}{k})$ expander graph, where $\frac{M(k)}{k}$ is its expansion factor~\cite{GUV07}.


\section{File Size of FR Codes and Generalized Hamming Weights}
\label{sec:genHamming}

Let $C$ be a $(\theta,k)$ linear code and $A$ be a subcode of $C$.
The \emph{support} of $A$, denoted  by $\chi(A)$, is defined by
$$\chi(A)\deff\{i:\exists(c_1,c_2,\ldots,c_{\theta})=\textbf{c}\in A,c_i\neq 0\}.
$$

The $r$th \emph{generalized Hamming weight} of a linear code $C$, denoted by  $d_r(C)$ ($d_r$ in short), is
the minimum support of any $r$-dimensional subcode of $C$, $1\leq r\leq k$, namely,
$$d_r=d_r(C)\deff \min_{A}\{|\chi(A)|: A \subseteq C, \dim (A) = r\}.
$$
Clearly, $d_r\leq d_{r+1}$ for $1\leq r\leq k-1$.
The set $\{d_1,d_2,\ldots,d_k\}$ is called the \emph{generalized Hamming
weight hierarchy} of~$C$~\cite{Wei91}.

There are a few definitions of generalized Hamming
weights for nonlinear codes~\cite{CLZ94,EtVa98,ReBe99}.
We propose now another straightforward definition for generalized
Hamming weight hierarchy for nonlinear codes. This definition is strongly
connected to the file sizes for different values of $k$ of a given FR code $C$.

Let $C$ be
a code of length $\theta$ with $n$ codewords. Assume further that
the all-zero vector is not a codeword of $C$ (if the all-zero vector
is a codeword of $C$ we omit it from the code). The $r$th generalized
Hamming weight of $C$, $d_r(C)$, will be defined as the minimum support of any
subcode of $C$ with $r$ codewords, i.e.,

$$d_r=d_r(C)\deff \min_{A}\{|\chi(A)|: A \subseteq C, |A| = r\}.
$$

Note that an $(n,\alpha,\rho)$
FR code $C$ can be represented as a binary constant weight code $C$ of length $\theta$ and weight $\alpha$.  Note
further that the minimum Hamming distance of $C$ is $2\alpha-2$.
Finally note that with these definitions we have that
$d_k=M(k)$. Therefore, by our previous discussion and
the definition of the generalized Hamming weight hierarchy
it is natural to define the \emph{file size hierarchy} of
an FR code $C$ to be the same as the generalized Hamming weight hierarchy of the related binary constant weight code $C$.

In addition to the questions discussed in the previous sections, the
definition of the file size hierarchy raises some natural questions.

\begin{enumerate}
\item Do there exist two FR codes
$C_1$ and $C_2$, with the same parameters $n,\alpha,\rho$,
 and two integers $k_1$ and $k_2$, such that $M_1(k_1)<M_2(k_1)$ and $M_1(k_2)>M_2(k_2)$,
 where $M_1(k)$ and $M_2(k)$ are the file sizes of $C_1, C_2$, respectively?

\item Given
$n,\alpha,\rho$ and a file size hierarchy $\{d_1=\alpha, d_2, \ldots, d_{\alpha}\}$, does there exist an FR code $C$ with these parameters,
which satisfies for each $k \leq \alpha$ that
$d_k=M(k)$ for every $k\leq \alpha$?
\end{enumerate}


\section{Bound on Reconstruction Degree}
\label{sec:bound_k}

In this section we consider a lower bound on the reconstruction degree $k$ for an FR code, given $M, \theta, n$, and $\alpha$.
Note, that given the value of file size $M$ it is desirable to have a reconstruction degree $k$ as small as possible, to provide the maximum possible failure resilience for the related DSSs. Hence, it is of interest to obtain FR codes which attain this bound.

\begin{lemma}
\label{lm:k-bound}
Let $C$ be an $(n,\alpha,\rho)$ FR code which stores a file of a given size $M$. The reconstruction degree $k$ of the corresponding system should satisfy
\begin{equation}
\label{eq:k-bound}
k\geq \left\lceil\frac{n\binom{M-1}{\alpha}}{\binom{\theta}{\alpha}}\right\rceil+1.
\end{equation}

\end{lemma}

\begin{proof}
Let  $X$ be an $n\times \binom{\theta}{M-1}$ binary matrix whose rows  are indexed by  $\{N_i\}_{i=1}^n$ of $C$  and  whose columns are indexed by all the possible $(M-1)$-subsets of $[\theta]$. The value $X_{i,j}$ is one if and only if $N_i$ is a subset of the $(M-1)$-set representing the
 $j$th column.  We count the number of ones in the matrix $X$ in two different ways, to obtain a lower bound on the value of the reconstruction degree $k$.

Since $|N_i|=\alpha$ and $N_i\subseteq [\theta]$, it follows that there are $\binom{\theta-\alpha}{M-1-\alpha}$ ways to complete an $\alpha$-subset of $[\theta]$ to an $(M-1)$-subset of $[\theta]$. Hence the number of ones in each row of $X$ is given by $\binom{\theta-\alpha}{M-1-\alpha}$. The number of ones in each column of $X$ is at most $(k-1)$, since otherwise, there exist $k$ subsets $N_{i_1},\ldots,N_{i_k}$ such that $|\cup_{s=1}^{k}N_{i_s}|\leq M-1$ which contradicts to the data reconstruction of FR codes, in other words, it is not possible to reconstruct the file from the $k$ nodes indexed by $i_1,\ldots, i_k$. Hence,
\begin{equation*}
n\binom{\theta-\alpha}{M-1-\alpha}\leq \binom{\theta}{M-1}(k-1),
\end{equation*}
which can be rewritten as
\[
\frac{n}{k-1}\leq \frac{\binom{\theta}{\alpha}}{\binom{M-1}{\alpha}},
\]
and the lemma follows.
\end{proof}

Two families of FR codes which  attain the bound in~(\ref{eq:k-bound}) are presented in the following two lemmas. The corresponding codes correct two node erasures, i.e., $n-k=2$, and the data/storage ratio is $\frac{\theta-1}{2\theta}$, i.e., almost 1/2.

\begin{lemma}
\label{lm:almost_complete_k}
Let $n>2$ be is an even integer and  let $K_n^-$ be an $(n-2)$-regular graph obtained by removing a perfect matching from the set of edges in $K_n$. Then the FR code $C_{K_n^-}$ based on the graph $K_n^-$ attains the bound in~(\ref{eq:k-bound}) for $k=n-2$.

\end{lemma}
\begin{proof}
First, observe that a perfect matching in $K_n$ is of size $\frac{n}{2}$ and hence the number of edges in $K_n^-$  is $\theta=\binom{n}{2}-\frac{n}{2}$. Any  $n-2$ vertices of $K_n^-$ are incident with at least $\theta-1$ edges and therefore $M(n-2)=\theta-1$. To prove that $C_{K_n^-}$  attains the bound in~(\ref{eq:k-bound}) for $k=n-2$ we have to prove that
\begin{equation*}
n-3=\left\lceil\frac{n\binom{\theta-2}{\alpha}}{\binom{\theta}{\alpha}}\right\rceil
\end{equation*}
for $\alpha=n-2$.
Since
\[\frac{n\binom{\theta-2}{\alpha}}{\binom{\theta}{\alpha}}=\frac{n(\theta-\alpha)(\theta-1-\alpha)}{(\theta-1)\theta}=\frac{(n-2)((n-2)^2-2)}{n^2-2n-2}=
n-4+\frac{4n-12}{n^2-2n-2},
\]
 the statement of the lemma is proved.
\end{proof}

Similarly to Lemma~\ref{lm:almost_complete_k} one can prove the following lemma.
\begin{lemma}
\label{lm:complete_k}
 The FR code $C_{K_n}$ based on $K_n$ attains the bound in~(\ref{eq:k-bound}) for $k=n-2$.
\end{lemma}

\begin{remark}$C_{K_n}$ was first defined in~\cite{Rashmi09,RSKR12}).  Note that this code is also an MBR code for which $n=\alpha+1$.
\end{remark}

As we already noticed, when all the other four  parameters $M, \theta, \alpha$ and $n$  (or, equivalently, $M, \rho, \alpha$ and $n$) are fixed, it is desirable to have the smallest possible reconstruction degree $k$. However, when only three parameters $M, \rho$ and  $\alpha$ are fixed, we have a trade-off between the reconstruction degree $k$ and the storage overhead of the DSS. We illustrate this trade-off by the following example.
Suppose that we want to store a file of size $M=36$  by using an FR code
with $\alpha=8$ and $\rho=2$. Note that $n$ is not fixed.
First, for $k=5$ we consider a $(114,8,2)$ FR code based
on a projective plane $PG(2,7)$ (see Example~\ref{ex:MooreFR}). Recall that this code is the optimal code, i.e., it stores a file of the maximum size since the corresponding graph has girth 6 which is  greater than $k=5$ (see Corollary~\ref{cor:girthOptimal}), and in addition, it has the minimum possible value of $n$, as the corresponding graph is a Moore graph. However,
the length of a corresponding MDS codeword and then the total storage is 456 and hence the storage overhead
is much more than 1000$\%$. To have smaller overhead one can use
a $(10,5)$-Tur\'an graph which for $k=7$ by Theorem~\ref{trm:Turan} yields a $(10,8,2)$ FR code
for which we can take a file of size $M=37$ and encode it with
a $(40,37)$ MDS code. Hence, the total overhead is only about 10$\%$ and moreover, the field size for the FR code is required to be much smaller than in the previous code.
 This example illustrates the fact that if we are given the file size $M$,
the number of symbols $\alpha$ stored in a node, and the repetition degree $\rho$, then decreasing the reconstruction degree $k$ increases the storage overhead significantly.

\section{Fractional Repetition Batch Codes}
\label{sec:batch}

In this section we analyze additional properties of DRESS codes (FR codes for which the reconstruction degree $k$ is determined; see Section~\ref{sec:introduction}) which
allow  also load balancing between storage nodes
by establishing a connection to combinatorial batch codes. We consider a scenario when in addition to the uncoded exact repairs of failed nodes and to the recoverability of the stored file  from any set of $k$ nodes we require an additional property. Given a positive  integer~$t$, any $t$-subset of  stored symbols can be retrieved by reading at most one symbol from each node. This retrieval can be performed by $t$ different users in parallel, where each user gets a different symbol.  In other words, we propose a new type of codes for DSS, called in the sequel \emph{fractional repetition batch}  (FRB) codes, which enable uncoded efficient node repairs and load balancing which is performed by several users in parallel. An FRB code is a combination of an FR code and an uniform combinatorial batch code.

The family of codes called \emph{batch codes}  was proposed for load balancing in distributed storage.
A \emph{batch code}, introduced in~\cite{IKOS04},  stores $\theta$ (encoded) data symbols on $n$ system nodes in such a way that any batch of $t$ data symbols can be decoded by reading at most one symbol from each node.
In a $\rho$-\emph{uniform} \emph{combinatorial batch code}, proposed in~\cite{PSW08},
each node stores a subset of the data symbols and no decoding is required during retrieval of any batch of $t$ symbols. Each symbol is stored in exactly $\rho$ nodes and hence it is also called a replication based batch code. A $\rho$-uniform combinatorial batch code is denoted by $\rho-(\theta,N,t,n)$-CBC, where $N=\rho \theta$ is the  total storage over all the $n$ nodes. These codes were studied in~\cite{IKOS04,PSW08,BRR11,BaBh12,BuTu13,SiGa13}.

Next, we provide a formal definition of FRB codes. This definition is based on the definitions of a DRESS code and a uniform combinatorial batch code.
Let $\bf f$ $\in \F_q^M$ be a file of size $M$ and let $c_{\bf f}\in \F_q^{\theta}$ be a codeword of an $(\theta,M)$ MDS code which encodes the data $\bf f$.
Let $\{N_1,\ldots,N_n\}$ be a collection of $\alpha$-subsets of the set $[\theta]$.
 A $\rho-(n,M,k,\alpha,t)$ \emph{FRB code} $C$, $k\leq \alpha$, $t\leq M$, represents a system of $n$ nodes with the following properties:
\begin{enumerate}
  \item Every node $i$, $1\leq i\leq n$, stores $\alpha$ symbols of $c_{\bf f}$ indexed by $N_i$;
  \item Every symbol of $c_{\bf f}$ is stored on $\rho$ nodes;
  \item From any set of $k$ nodes it is possible to reconstruct the stored file $\bf f$, in other words, $M=\min_{|I|=k}|\cup_{i\in I}N_i|$;
  \item Any batch of $t$ symbols from $c_{\bf f}$ can be retrieved by downloading at most one symbol from each node.
\end{enumerate}
Note that the total storage over the $n$ nodes needed to store the file $\bf f$ equals to $n\alpha=\theta\rho$.

\begin{remark}
Note that while in a classical batch code any $t$ \emph{data} symbols can be retrieved, in an FRB code any batch of $t$ \emph{coded} symbols can be retrieved.  In particular, when a systematic MDS code is chosen for an FRB code, the data symbols can be easily retrieved.
\end{remark}

To present constructions of FRB codes, we need the following results on constructions of uniform combinatorial batch codes.

\begin{theorem}~\cite{PSW08}
\label{thm:batch_girth}
Let $G$  be a graph with $n$ vertices, $\theta$ edges and girth $g$. Then  the batch code  $C^B_{G}$ with nodes indexed by the vertices of $G$ and with  data symbols indexed by the edges of $G$, is a $2-(\theta, 2\theta,t,n)$-CBC with $t=2g-\lfloor g/2\rfloor-1$.
\end{theorem}

\begin{theorem}~\cite{SiGa13}
\label{thm:batch_TD}
Let TD be a resolvable transversal design $\textmd{TD}(q-1,q)$, for a prime power $q$. Then the  batch code  $C^B_{\textmd{TD}}$ with  nodes indexed by  points and data symbols indexed by blocks of TD, is a $(q-1)-(q^2,q^3-q^2,q^2-q-1,q^2-q)$-CBC.
\end{theorem}

By applying  Theorems~\ref{thm:batch_girth}, Theorem~\ref{thm:batch_TD} together with Corollary~\ref{thm:TD(2,a)}, Theorem~\ref{thm:rateGirth}, and Theorem~\ref{lm:TDrate} we obtain the following result.
\begin{theorem}
\label{thm:FRB from TD}
$~$
\begin{enumerate}
  \item Let $K_{\alpha,\alpha}$ be a complete bipartite graph with $\alpha>2$. Then $C_{K_{\alpha,\alpha}}$ is a $2-(2\alpha, M,k,\alpha,5)$ FRB code with
  $M=M(k)=k\alpha-\left\lfloor\frac{k^2}{4}\right\rfloor$.
  \item Let $G$ be an $\alpha$-regular graph on $n$ vertices with girth $g$. Then $C_G$ is a $2-(n, M, k,\alpha,2g-\lfloor g/2\rfloor-1)$ FRB code with
  \[M=M(k)=\left\{\begin{array}{cc}
         k\alpha -k+1 \;&\textmd{ if } k\leq g-1 \\
           k \alpha-k \;&\textmd{ if }g\leq k\leq g+\lceil\frac{g}{2}\rceil-2.
         \end{array}\right.
\]
  \item Let TD be a resolvable transversal design $\textmd{TD}(\alpha-1,\alpha)$,  for a prime power $\alpha$. Then $C_{\textmd{TD}}$ is an $(\alpha-1)-(\alpha^2-\alpha, M,k,\alpha,\alpha^2-\alpha-1)$ FRB code with
  $M\geq k\alpha -\binom{k}{2}+(\alpha-1)\binom{x}{2}+xy$, where $x,y$ are nonnegative integers which satisfy $k=x(\alpha-1)+y$, $y\leq \alpha-2$.
\end{enumerate}
\end{theorem}

\begin{example}
{~}
\begin{itemize}
\item Consider the FRB code $C_{K_{3,3}}$ based on $K_{3,3}$ (see also Example~\ref{ex:bipartite} for an FR code based on $K_{3,3}$).
By Theorem~\ref{thm:FRB from TD}, for $k=3$, $C_{K_{3,3}}$ is a $2-(6, 7,3,3,5)$ FRB code.

\item Consider the FRB code $C_{\textmd{TD}}$  based on the resolvable transversal design $\textmd{TD}=\textmd{TD}(3,4)$ (see also Example~\ref{ex:TD_34} for an FR code based on $\textmd{TD}(3,4)$). By Theorem~\ref{thm:FRB from TD}, for $k=4$, $C_{\textmd{TD}}$ is a $3-(12, 11,4,4,11)$ FRB code, which stores a file of size $11$ and allows for retrieval of  any (coded) $11$ symbols, by reading at most one symbol from a node. In particular, when using a systematic MDS code, $C_{\textmd{TD}}$  provides load balancing in data reconstruction.
\end{itemize}

\end{example}

For the rest of this section we consider FRB codes obtained from affine planes. Uniform combinatorial batch codes based on affine planes were considered in~\cite{SiGa13}.

 An \emph{affine plane} of order $s$, denoted by $A(s)$, is  a design $(\cP,\mathcal{B})$, where $\cP$ is a set of $|\cP|=s^2$ \emph{points}, $\mathcal{B}$ is a a collection of $s$-subsets (\emph{blocks}) of $\cP$ of size $|\mathcal{B}|=s(s+1)$, such that each pair of points in $\cP$ occur together in exactly one block of~$\mathcal{B}$.
 An affine plane is called \emph{resolvable}, if the \emph{set} $\mathcal{B}$ can be partitioned into $s+1$ sets of size~$s$, called parallel classes, such that every element of $\cP$ is contained in exactly one block of each class. It is well known~\cite{Anderson} that if $q$ is a prime power, then there exists a resolvable affine plane $A(q)$.

\begin{theorem}
\label{thm:affine}
Let  $A(q)$ be an affine plane, for a prime power $q$, and let $C_{A(q)}$ be an FRB based on $A(q)$, i.e., $\textbf{I}(C_{A(q)})=\textbf{I}(A(q))$.
Then $C_{A(q)}$  is a $q-(q^2, k(q+1)-\binom{k}{2},k,q+1,q^2)$ FRB code.
\end{theorem}
\begin{proof}
The parameters $\rho, n,\alpha$ and $t$ follow from the properties of the batch code based on $A(q)$ (see~\cite{SiGa13} for details).
Since any two points of $A(q)$ are contained in exactly one block and hence any two rows of $\textbf{I}(C_{A(q)})$ intersect, it follows that the file size is $k(q+1)-\binom{k}{2}$.
\end{proof}

\begin{remark}
The structure of an incidence matrix of an FRB code and some other constructions are considered in~\cite{Sil14}.
\end{remark}


\section{Conclusion}
\label{sec:conclusion}

We considered the problem of constructing  $(n,\alpha,\rho)$ FR codes which
attain the upper bound on the file size.
We presented constructions of FR codes based on
regular graphs, namely, Tur\'an graphs and
graphs with a given girth; and constructions of FR codes based  on combinatorial designs, namely, transversal designs and generalized polygons.
The problems of constructing  optimal FR codes and FR codes with a given file size raise  interesting  questions in graph theory (see Section~\ref{sec:rho2}
and Section~\ref{sec:rho>2}).
We  defined the file size hierarchy of FR codes, which is a possible definition of generalized Hamming weight hierarchy for nonlinear codes. For this we presented FR codes as  binary constant weight codes.
In addition, we derived a lower bound on the reconstruction degree for  FR codes and presented FR codes which attain this bound.
Finally, based on a connection between FR codes and batch codes, we proposed a new family of codes for DSS, namely fractional repetition batch codes, which have the properties of batch codes and FR codes simultaneously.  These are the first codes for DSS which allow for uncoded efficient repairs and load balancing. We
presented  examples of  constructions for FRB codes, based on combinatorial designs, complete bipartite graphs, and graphs with large girth.

In general, given four of the five parameters of FR codes, namely,
the number $n$ of nodes, the number $k$ of nodes
needed to reconstruct the whole stored file,
the number $\alpha$ of stored symbols in a node,
the number $\rho$ of repetitions of a symbol in the code,
and the size $M$ of the stored file $\textbf{f}$, one can ask
what are the possible values of the fifth parameter.
For this we  define the following five functions.

\begin{enumerate}
\item Let $n(k,\alpha,\rho,M)$ be the minimum number $n$
of nodes  for an $(n,\alpha,\rho)$ FR code which stores a file of size~$M$ for a given reconstruction degree $k$.

\item Let $k(n,\alpha,\rho,M)$ be the minimum number $k$
of nodes  from which the whole stored file of size $M$, of an $(n,\alpha,\rho)$ FR code, can be reconstructed.

\item Let $\alpha(n,k,\rho,M)$ be the minimum number $\alpha$ of symbols stored in a node of
 an $(n,\alpha,\rho)$ FR code which stores a file of size~$M$  for a given reconstruction degree $k$.

\item Let $\rho(n,k,\alpha,M)$ be the minimum number $\rho$
of repetitions of a symbol in
an $(n,\alpha,\rho)$ FR code which stores a file of size~$M$  for a given reconstruction degree $k$.

\item Let $M(n,k,\alpha, \rho)$ be the maximum size of a stored file
in an $(n,\alpha,\rho)$ FR code,  for a given reconstruction degree $k$.
\end{enumerate}

This formulation is very similar to classical coding theory, where for example three functions are defined for the trade-off between the length of the code, its size, and its minimum distance.
In this paper, we considered the values of three functions out of the five, namely, $n(k,\alpha,\rho,M)$, $k(n,\alpha,\rho,M)$ and $M(n,k,\alpha, \rho)$.


%
%
%
\appendices

\section{}
\label{app:triangle}
\textbf{Constructions for FR codes with $\bf{\emph{M}(3)=3\alpha-2}$:}
Let $\textbf{I}(K_{\alpha+i,\alpha+i})$, $i\geq 0$, be the $2(\alpha+i)\times (\alpha+i)^2$ incidence matrix of the complete bipartite graph $K_{\alpha+i,\alpha+i}$. Note that it is also the incidence matrix of a resolvable transversal design $\textmd{TD}(2,\alpha+i)$. Based on the resolvability of the design, $\textbf{I}(K_{\alpha+i,\alpha+i})$ can be written in a blocks form, i.e., $\textbf{I}(K_{\alpha+i,\alpha+i})$ is a $2\times (\alpha+i)$ blocks matrix, where each block is a permutation matrix of size $(\alpha+i)\times (\alpha+i)$. Each such permutation matrix block will be called a $p$\emph{-block}.
W.l.o.g. we assume that the first two $p$-blocks  which correspond to the first $\alpha+i$ columns of the matrix  are identity matrices.
Let $\textbf{I}_{\alpha,i}^{\textmd{even}}$ be a matrix obtained from $\textbf{I}(K_{\alpha+i,\alpha+i})$ by removing $i(\alpha+i)$ columns  which correspond to $2i$ $p$-blocks (w.l.o.g. we assume that we removed the leftmost columns).  Note that there are exactly $\alpha$ ones in each row of $\textbf{I}_{\alpha,i}^{\textmd{even}}$. Let $C_{\alpha,i}^{\textmd{even}}$ be an FR code obtained from the graph $G_{\alpha,i}^{\textmd{even}}$, whose incidence matrix is $\textbf{I}_{\alpha,i}^{\textmd{even}}$. It is easy to verify that $C_{\alpha,i}^{\textmd{even}}$ is a $(2\alpha+2i, \alpha, 2)$ code whose  file size is $M(3)=3\alpha-2$. Note that
the  data/storage  ratio $\frac{M(3)}{n\alpha}=\frac{3\alpha-2}{(2\alpha+2i)\alpha}$ decreases when $i$ increases.

For odd $n$, $n\geq \frac{5}{2}\alpha$, and even $\alpha$ we construct an $(n,\alpha, 2)$ FR code $C$ with $M(3)=3\alpha-2$ as follows.
We distinguish between two cases, odd $\alpha/2$ and even $\alpha/2$.

For odd $\alpha/2$ let $G=(V,E)$ be a graph whose vertex set is given by $V=\{X,Y_0,Y_1,Z_0,Z_1\}$, where $|X|=|Y_0|=|Y_1|=\alpha/2$, $|Z_0|=|Z_1|=\alpha/2+j$, for any given integer $j\geq 0$. The edges in $G$ are given by
\begin{enumerate}
  \item $E_1=\{\{v,u\}: v\in X,\;u\in Y_i, i=0,1\}$;
  \item $E_{2+\ell}=\{\{v_i,u_j\}: v_i\in Y_{\ell},\;u_j\in Z_{\ell},  i=1,2,\ldots, \frac{\alpha}{2}, j=((i-1)\frac{\alpha}{2}+1)\textmd{ mod }{|Z_{\ell}|},\ldots, i\frac{\alpha}{2}\textmd{ mod }{|Z_{\ell}|}\}$, $\ell=0,1$, such that the degree of a vertex $x$  in the induced  subgraph $(Y_{\ell}\cup Z_{\ell}, E_{2+\ell})$, is given by
      $$\textmd{deg}(x)=\left\{\begin{array}{cc}
                        \frac{\alpha}{2} & x\in Y_{\ell} \\
                        t:=\left\lfloor\frac{\alpha^2/4}{|Z_1|}\right\rfloor \textmd{ or } t+1 & x\in Z_{\ell}
                      \end{array}
                      \right.
      $$
  \item $E_4=\{\{v_i,u_i\}:v_i\in Z_0, u_i\in Z_1, \textmd{deg}(v_i)=\textmd{deg}(u_i)=t \textmd{ in } (Y_0\cup Z_0, E_2) \textmd{ and } (Y_1\cup Z_1, E_3)$
  \item $E_5=\{\{v,u\}:v\in Z_0, u\in Z_1, s.t. (Z_0\cup Z_1, E_5)=G^{\textmd{even}}_{(\alpha-t-1),(|Z_{\ell}|-\alpha+t+1)}\}$.
\end{enumerate}
The schematic graph $G$ is shown on Fig.~\ref{fig:graph}. Clearly, the degree of every vertex in $G$ is $\alpha$ and there are $\frac{5}{2}\alpha+2j$
vertices in $G$.

\begin{figure*}[t]
\centering
\includegraphics[width=0.5\textwidth]{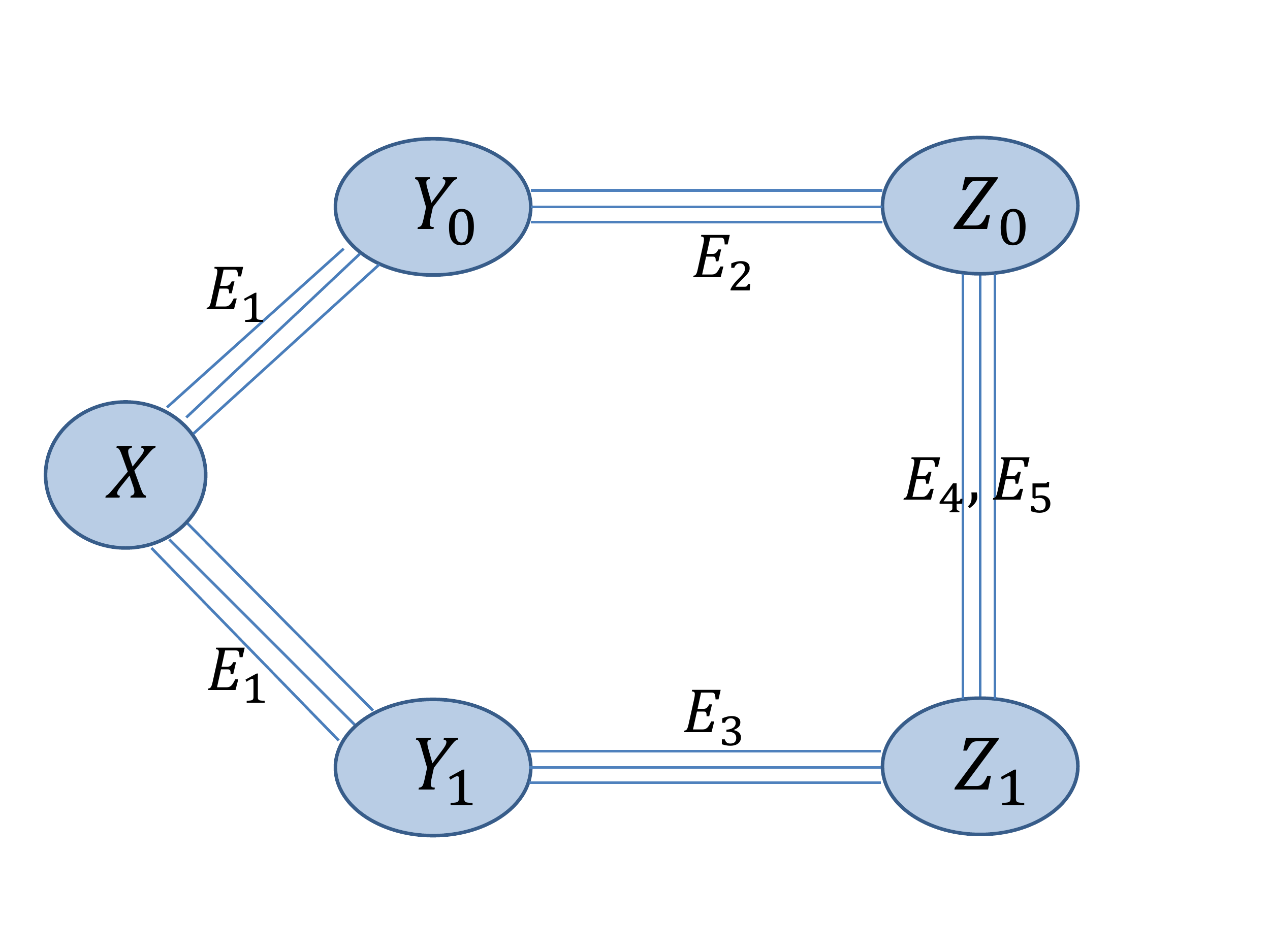}
\caption{A schematic structure of a graph without triangles, with odd number of vertices}\label{fig:graph}
\end{figure*}

For even $\alpha/2$, we define the graph $G$ in the similar way to the case in which $\alpha/2$ is odd. The only difference is that $|X|=\alpha/2-1$ and $|Z_0|=|Z_1|=\alpha/2+1+j$; $E_1, E_2,\ldots, E_5$ are defined similarly and the number of vertices in $G$ is $\frac{5}{2}\alpha+1+2j$.
We illustrate this construction with the following example.
\begin{example} For $\alpha=4$ the incidence  matrix $\textbf{I}(G)$ of the graph $G$ is given by
\begin{footnotesize}
\[\left(
    \begin{array}{cccc|cccccc|cccccc|cccccc}
      1 & 1 & 1 & 1 &  &  &  &  &  &  &  &  &  &  &  &  &  &  &  &  &  &  \\\hline
      1 &  &  &  & 1 & 1 & 1 &  &  &  &  &  &  &  &  &  &  &   &  &  &  &  \\
       & 1 &  &  &  &  &  & 1 & 1 & 1 &  &  &  &  &  &  &  &   &  &  &  &  \\\hline
       &  & 1 &  &  &  &  &   &  &  & 1 & 1 & 1 &  &  &  &  &  &  &  &  &  \\
       &  &  & 1 &  &  &  &   &  &  &  &  &  & 1 & 1 & 1 &  &  &  &  &  &  \\\hline
       &  &  &  & 1 &  &  & 1 &  &  &  &  &  &  &  &  & 1 &  &  & 1 &  &  \\
       &  &  &  &  & 1 &  &  & 1 &  &  &  &  &  &  &  &  & 1 &  &  & 1 &  \\
       &  &  &  &  &  & 1 &  &  & 1 &  &  &  &  &  &  &  &  & 1 &  &  & 1 \\\hline
       &  &  &  &  &  &  &  &  &  & 1 &  &  & 1 &  &  & 1 &  &  &  & 1 &  \\
       &  &  &  &  &  &  &  &  &  &  & 1 &  &  & 1 &  &  & 1 &  & 1 &  &  \\
       &  &  &  &  &  &  &  &  &  &  &  & 1 &  &  & 1 &  &  & 1 &  &  & 1 \\
    \end{array}
  \right),
\]
\end{footnotesize}
where an empty entry in the matrix is $0$ and the blocks of the matrix correspond to the partition of vertices and edges as described in the construction.
\end{example}

\emph{Proof of Lemma}~\ref{lm:odd-triangle}:
\label{app:lm-triangle}
The upper bound on $n_3$ directly follows from the discussion in the beginning of Subsection~\ref{subsec:file_size}. To prove the lower bound, we need to show that there is no $\alpha$-regular graph with $2\alpha+1$ vertices which does not contain a cycle of length~3. Let $G=(V,E)$ be an $\alpha$-regular graph with $2\alpha+1$ vertices. Let $v\in V$ and let $A=\{a_1,a_2,\ldots,a_{\alpha}\}$ be the set of its neighbours.
Since there are no cycles of length 3, there are no edges between vertices of $A$.
Let $B=\{b_1,b_2,\ldots, b_{\alpha}\}$ be the remaining $\alpha$ vertices of $G$.
Clearly, every vertex of $A$ has $\alpha-1$ neighbours in $B$.
%
%
The number of edges in the cut between $A$ and $B$ is $\alpha(\alpha-1)$ and hence the number of edges between the vertices in $B$ is $\frac{\alpha^2-(\alpha(\alpha-1))}{2}=\frac{\alpha}{2}$. W.l.o.g. let $\{a_1, b_1\}\notin E$ and $\{a_2, b_1\}\in E$. The vertices of $B\setminus\{b_1\}$ are neighbours of $a_1$, which implies  that there are no edges between  vertices in $B\setminus\{b_1\}$. Therefore, $b_1$ has $\frac{\alpha}{2}$ neighbours in $B\setminus\{b_1\}$ and $a_2$ has $\alpha-2$ neighbours in $B\setminus\{b_1\}$. This implies that if $\alpha>2$ then there exists a common neighbour $b_{\ell}$ of $a_2$ and $b_1$ in $B\setminus\{b_1\}$. Hence, $a_2,b_1$ and $b_{\ell}$ form a cycle, a contradiction.

\hfill $\blacksquare$

\begin{example}
\label{ex:ap-motdification}
Let $\alpha=7$ and $M(4)=4\alpha-5=23$. By the Tur\'an bound (see Corollary~\ref{cor:Turan}), the minimum number of nodes in an FR code is
$\left\lceil\frac{3}{2}\alpha\right\rceil=11$, but since $n\alpha$ must to be an even number, it follows that $n\geq 12$.  The corresponding Tur\'an graph  $T(12,3)$  with $n=12$ has degree $8$. To obtain from $T(12,3)$ a $7$-regular graph on 12 vertices without $G(4,6)$,  one can remove 6 edges, which correspond to a perfect matching in $T(12,3)$.

\end{example}

\section{}

\emph{Proof of Theorem~\ref{thm:alphaBound}}:
\label{app:thm:alphaBound}
We observe that if the lower bound on the file size of Theorem~\ref{lm:TDrate} is equal to the upper bound in~(\ref{eq:bound2}) then we have
\begin{equation}
\label{eq:TD_induction}M(k)= k\alpha-\binom{k}{2}+\rho\binom{b}{2}+bt=\varphi(k).
 \end{equation}
We prove~(\ref{eq:TD_induction}) by induction.
We distinguish between two cases: $k=b\rho+t$, where $t\neq 0$, i.e., $k\not\equiv 0 \;(\textmd{mod }\rho)$  and $k=b\rho$, i.e., $k\equiv 0\;(\textmd{mod }\rho)$.

\emph{Case 1:} if $k=b\rho+t$, $0<t\leq \rho-1$, then $k-1=b\rho+t-1$, where $0\leq t-1\leq \rho-2$.
Assume that $\varphi(k-1)=M(k-1)=(k-1)\alpha-\binom{k-1}{2}+\rho\binom{b}{2}+b(t-1)$. By applying the recursive formula for $\varphi$, we want to prove that
\begin{equation}
\label{eq:TD_proof}
k\alpha-\binom{k}{2}+\rho\binom{b}{2}+bt=
(k-1)\alpha-\binom{k-1}{2}+\rho\binom{b}{2}+b(t-1)+\alpha-\left\lceil\frac{(\rho-1)(k-1)\alpha-\rho\binom{k-1}{2}+\rho^2\binom{b}{2}+\rho b(t-1)}{\rho\alpha-k+1}\right\rceil
\end{equation}
By simplifying the last equation one can verify that~(\ref{eq:TD_proof}) holds if and only if
\begin{equation}
\label{eq:TD_proof2}
(\rho-1)(k-1)\alpha-\rho\binom{k-1}{2}+\rho^2\binom{b}{2}+\rho b(t-1)>(k-b-2)(\rho\alpha-k+1).
\end{equation}
By substituting $k=b\rho+t$ in~(\ref{eq:TD_proof2}) we have that~(\ref{eq:TD_proof2}) holds if and only if
\[\alpha >\frac{ b^2\rho\binom{\rho-1}{2}+(\rho-2)\binom{t-1}{2}+b((\rho^2+1)(t-1)-\rho(3t-4))}{\rho-t+1}.
\]

\emph{Case 2:}   if $k=b\rho$ then $k-1=(b-1)\rho+\rho-1$.
Assume that $\varphi(k-1)=M(k-1)=(k-1)\alpha-\binom{k-1}{2}+\rho\binom{b-1}{2}+b(\rho-1)$. By applying the recursive formula for $\varphi$, we want to prove that
\begin{equation}
\label{eq:TD_proof3}
k\alpha-\binom{k}{2}+\rho\binom{b}{2}=(k-1)\alpha-\binom{k-1}{2}+\rho\binom{b-1}{2}+b(\rho-1)+\alpha-
\left\lceil\frac{(\rho-1)(k-1)\alpha-\rho\binom{k-1}{2}+\rho^2\binom{b-1}{2}+b\rho(\rho-1)}{\rho\alpha-k+1}\right\rceil.
\end{equation}

One can verify that~(\ref{eq:TD_proof3}) holds if and only if
\begin{equation}
\label{eq:TD_proof4}
(\rho-1)(k-1)\alpha-\rho\binom{k-1}{2}+\rho^2\binom{b-1}{2}+b\rho(\rho-1)>(k-b-1)(\rho\alpha-k+1).
\end{equation}
By substituting $k=b\rho$ in~(\ref{eq:TD_proof4}) we have that~(\ref{eq:TD_proof4}) holds if and only if
\[\alpha >\frac{b^2\rho}{2}(\rho^2-3\rho+2)-b(\rho^2-3\rho+1)-1=b(b\rho-2)\binom{\rho-1}{2}+b-1.
\]
\hfill $\blacksquare$

\end{document}